\documentclass[10pt]{article}

\usepackage{amsfonts}
\usepackage{amsmath}
\usepackage{amssymb}
\usepackage{amsthm}
\usepackage{latexsym}
\usepackage{mathtools}
\usepackage{mathrsfs}
\usepackage{mathpazo}
\usepackage[bb=dsserif]{mathalpha}
\usepackage{bm}

\usepackage{soul}
\usepackage{comment}
\usepackage{xargs}

\usepackage[pdftex,dvipsnames]{xcolor}
\usepackage[
    colorinlistoftodos,
    prependcaption,
    textsize=tiny
]{todonotes}

\usepackage[margin=1in]{geometry}
\usepackage[colorlinks=true, linkcolor=blue, urlcolor=blue, citecolor=blue]{hyperref}
\usepackage{aliascnt}
\usepackage{cleveref}
\usepackage{array}
\usepackage{cite}

\usepackage{authblk}

\usepackage{tikz}
\usetikzlibrary{calc,positioning}
\usepackage{mdframed}

\hypersetup{pdfpagemode=UseNone}
\newcommand{\tikzrefsize}[1]{\ifmmode{\scriptscriptstyle{#1}}\else{\tiny{#1}}\fi}
\newtheorem{theorem}{Theorem}

\newaliascnt{lemma}{theorem}
\newtheorem{lemma}[lemma]{Lemma}
\aliascntresetthe{lemma}
\crefname{lemma}{Lemma}{Lemmas}

\newaliascnt{prop}{theorem}
\newtheorem{prop}[prop]{Proposition}
\aliascntresetthe{prop}
\crefname{prop}{Proposition}{Propositions}

\newaliascnt{corollary}{theorem}

\aliascntresetthe{corollary}
\crefname{corollary}{Corollary}{Corollaries}

\theoremstyle{definition}

\newaliascnt{definition}{theorem}

\aliascntresetthe{definition}
\crefname{definition}{Definition}{Definitions}

\newaliascnt{remark}{theorem}
\newtheorem{remark}[remark]{Remark}
\aliascntresetthe{remark}
\crefname{remark}{Remark}{Remarks}

\newaliascnt{example}{theorem}

\aliascntresetthe{example}
\crefname{example}{Example}{Examples}

\mdfdefinestyle{figstyle}{
    linecolor=black!7,
    backgroundcolor=black!7,
    innertopmargin=10pt,
    innerleftmargin=25pt,
    innerrightmargin=25pt,
    innerbottommargin=10pt
}

\definecolor{White}{rgb}{1,1,1}
\definecolor{Black}{rgb}{0,0,0}
\definecolor{LightGray}{rgb}{.81,.81,.81}
\definecolor{Ali}{rgb}{0.68, 0.18, 0.08}

\colorlet{ChannelColor}{LightGray}
\colorlet{ChannelTextColor}{Black}
\colorlet{ReadoutColor}{White}

\newcommand{\eqdef}{\overset{\mathrm{def}}{=}}
\DeclareMathOperator{\tr}{Tr}
\DeclareMathOperator{\Tr}{Tr}
\DeclareMathOperator{\supp}{supp}
\DeclareMathOperator{\id}{id}

\DeclareMathOperator{\bounded}{\mathcal{B}}
\DeclareMathOperator{\wt}{\mathsf{wt}}
\DeclareMathOperator{\Span}{span}

\newcommand{\pauli}[2]{\mathsf{\MakeUppercase{#1}}_{#2}}
\newcommand{\NP}{\mathsf{NP}}
\newcommand{\QMA}{\mathsf{QMA}}

\NewDocumentCommand{\Pn}{o}{
  \mathcal{P}_{\!n\IfValueT{#1}{,#1}}
}
\newcommand{\Pcc}{\mathscr{P}_n}

\newcounter{propitem}
\newcommand{\Prop}[1]{\textup{\textbf{(\ref{#1})}}}
\renewcommand{\thepropitem}{\textbf{P\arabic{propitem}}}
\newcommand{\pitem}{%
    \refstepcounter{propitem}%
    \item[\textup{\textbf{(\thepropitem)}}]%
}

\newcommand{\complex}{\mathbb{C}}
\newcommand{\real}{\mathbb{R}}

\newcommand{\I}{\mathbf{1}}
\renewcommand{\natural}{\mathbb{N}}

\newcommand{\X}{\mathbb{X}}
\newcommand{\Y}{\mathbb{Y}}
\newcommand{\Z}{\mathbb{Z}}
\newcommand{\Iq}{\mathbb{1}}

\newcommand{\Ac}{\mathcal{A}}
\newcommand{\Bc}{\mathcal{B}}

\newcommand{\Ec}{\mathcal{E}}

\newcommand{\Hc}{\mathcal{H}}
\newcommand{\Kc}{\mathcal{K}}

\newcommand{\Pc}{\mathcal{P}}

\newcommand{\Xc}{\mathcal{X}}

\newcommand{\Pb}{\bm{\mathrm{P}}}
\newcommand{\Qb}{\bm{\mathrm{Q}}}

\newcommand{\ib}{\bm{i}}

\newcommand{\blambda}{\bm{\lambda}}
\newcommand{\bbeta}{\bm{\eta}}

\newcommand{\abs}[1]{\left|#1\right|}
\newcommand{\norm}[1]{\left\|#1\right\|}

\newcommand{\ceil}[1]{\left\lceil #1 \right\rceil}

\newcommand{\inter}[2]{\left[#1:#2 \right]}

\newcommand{\ket}[1]{\left| #1 \right\rangle}
\newcommand{\bra}[1]{\left\langle #1 \right|}
\newcommand{\ketbra}[2]{\left|#1\middle\rangle\!\middle\langle#2\right|}

\newcommand{\choi}[1]{J_{#1}}
\renewcommand{\vec}[1]{\lvert #1 \rangle\!\!\rangle}
\newcommand{\cev}[1]{\langle\!\!\langle #1 \rvert}

\newcommand{\cevvec}[2]{\langle\!\!\langle #1 \!\mid\! #2 \rangle\!\!\rangle}

\newcommandx{\unsure}[2][1=]{%
    \todo[
        linecolor=red,
        backgroundcolor=red!25,
        bordercolor=red,
        #1
    ]{#2}%
}

\newcommandx{\change}[2][1=]{%
    \todo[
        linecolor=blue,
        backgroundcolor=blue!25,
        bordercolor=blue,
        #1
    ]{#2}%
}

\newcommandx{\info}[2][1=]{%
    \todo[
        linecolor=OliveGreen,
        backgroundcolor=OliveGreen!25,
        bordercolor=OliveGreen,
        #1
    ]{#2}%
}

\newcommandx{\improvement}[2][1=]{%
    \todo[
        linecolor=Plum,
        backgroundcolor=Plum!25,
        bordercolor=Plum,
        #1
    ]{#2}%
}

\newcommandx{\thiswillnotshow}[2][1=]{%
    \todo[disable,#1]{#2}%
}

\makeatletter
\let\@fnsymbol\@arabic
\makeatother

\newcommand\blfootnote[1]{%
  \begingroup
  \renewcommand\thefootnote{}\footnote{#1}%
  \addtocounter{footnote}{-1}%
  \endgroup
}

\begin{document}

\title{\LARGE\bf Convergence rates of Sum-of-Hermitian-Squares Hierarchies for the Pauli algebra}

\author{Ali Almasi}
\author{D\'avid Bug\'ar}
\author{Cambyse Rouz\'e}
\author{Peter Brown}
  
\affil{{Télécom Paris, Inria, LTCI, Institut Polytechnique de Paris, 19 Place Marguerite Perey, 91120 Palaiseau, France}}

\date{\today}

\renewcommand\Affilfont{\normalsize\itshape}
\renewcommand\Authfont{\large}
\setlength{\affilsep}{6mm}
\renewcommand\Authand{\rule{10mm}{0mm}}

\maketitle

\begin{abstract}
  Moment/Sum-of-Hermitian-Squares relaxations for noncommutative polynomial optimization problems have become an important tool for analysing problems within quantum theory. Despite their widespread success, little is known about their rate of convergence and, consequently, their accuracy. In this work, we develop explicit convergence rates for relaxations of noncommutative polynomial optimization problems generated from the Pauli algebra -- covering applications to the ground state energy problem for $n$-qubit systems. In particular, we show that the rate of convergence can be bounded in terms of the smallest roots of a family of orthogonal polynomials known as Krawtchouk polynomials. Our result represents the first quantitative analysis of the rate of convergence for relaxations of noncommutative polynomial optimization problems. 
\end{abstract}

%------------------------------------------------------------------------------%
\section{Introduction}
\blfootnote{During the final preparation of the manuscript we became aware of a related work by Klep \textit{et al}, titled ``Quantitative semidefinite certificates for ground-state energies of Pauli Hamiltonians''~\cite{klep2026quantitative}, which obtains similar results. The two works were carried out independently.}
Given a Hamiltonian $H$ defined on a system of $n$-qubits, the \emph{ground state energy problem} asks to compute the minimum eigenvalue of $H$, which can equivalently be formulated as
\begin{equation}
\label{eq:gs_intro}
    H_{\min} \eqdef 
    \sup
    \left\{
        \lambda \in \real : H - \lambda \I \succeq 0
    \right\},
\end{equation}
where 
$\I$
is the identity operator on 
$n$ qubits. Estimating the ground state energy of a quantum system is a fundamental task in quantum theory. In condensed matter physics, determining the ground state energy is used to study quantum phase transition phenomena at zero temperature \cite{sachdevQuantumPhaseTransitions2011} and in quantum chemistry, estimates of the lowest energy of systems can be used to predict molecular stability, reaction rates and other chemical properties \cite{szaboModernQuantumChemistry1996}.
One groundbreaking result in quantum complexity theory is that for general $d$-local Hamiltonians with $d \ge 2$, approximating the ground state energy within an inverse polynomial precision is complete for the Quantum-Merlin-Arthur ($\QMA$) complexity class \cite{kitaevClassicalQuantumComputation2002, kempeComplexityLocalHamiltonian2006}, and similar hardness results have been established for several restricted classes of local Hamiltonians \cite{kayQuantumMerlinArthurcompleteTranslationallyInvariant2007, oliveiraComplexityQuantumSpin2008, aharonovPowerQuantumSystems2009a, huangTwodimensionalLocalHamiltonian2021}, including the model used for defining the quantum max-cut problem \cite{cubittComplexityClassificationLocal2016, piddockComplexityAntiferromagneticInteractions2017}, which has recently attracted significant interest \cite{gharibianAlmostOptimalClassical2019, anshu_et_alBeyondProductStateApproximations,BravyiApproximationAlgorithmsForQuantumMany-bodyProblems, hallgren_et_alAnApproximationAlgorithmForTheMAX-2-LocalHamiltonianProblem, parekh_et_al:BeatingRandomAssignmentForApproximatingQuantum2-LocalHamiltonianProblems,
parekhApplicationLevel2Quantum2021,
parekh2026optimalproductstateapproximation2local,
takahashi2026su2symmetricsemidefiniteprogramminghierarchy,
King2023improved,
Watts2024relaxationsexact,
huberSecondOrderConeRelaxationsForQuantumMaxCut,
gribling2025improvedapproximationratiosquantum,
ju_et_al:ImprovedApproximationAlgorithmsForTheEPRHamiltonian,
apte_et_al:ImprovedAlgorithmsForQuantumMaxCutViaPartiallyEntangledMatchings,
huber2025lovaszthetalowerbound
}.

Despite these hardness results, a range of methods have been introduced to approximate the ground state energy \cite{verstraeteProjectedEntangledStates2006,verstraeteMatrixProductStates2008,austinQuantumMonteCarlo2012,langeArchitecturesApplicationsReview2024}. One such method is based on formulating the ground state problem as an instance of a class of optimization problems known as \emph{noncommutative polynomial optimization} (NPO) problems
\cite{navascuesConvergentHierarchySemidefinite2008,pironioConvergentRelaxationsPolynomial2010}, which generalize \emph{polynomial optimization} in the commutative setting \cite{lasserreGlobalOptimizationPolynomials2001, parriloStructuredSemidefinitePrograms2000}. NPO problems take a general form
as
\begin{equation}
    \label{eq:intro_npo}
\begin{aligned}
\inf_{\mathcal{H},\, X,\, \ket{\psi}} \quad & \bra{\psi} f(X) \ket{\psi} \\
\text{s.t.} \quad & p_i(X) \succeq 0, \quad i = 1, \dots, m_p,\\
& q_j(X) \ket{\psi} = 0, \quad j = 1, \dots, m_q,\\
& \bra{\psi} r_k(X) \ket{\psi} \geq 0, \quad k=1, \dots, m_r,
\end{aligned}
\end{equation}
where 
$f$,
$p_i$,
$q_j$
and
$r_k$ 
are polynomials in the non-commuting variables 
$(X_1, \dots, X_n)$, with 
$f$, $p_i$ and $r_k$ additionally being Hermitian.
The optimization is over all separable Hilbert spaces
$\Hc$,
all tuples
$(X_1, \dots, X_n)$
of bounded operators on $\Hc$
and all unit vectors
$\ket{\psi}$ in 
$\Hc$. 

From the fact that the ground state energy problem is an instance of an NPO problem, we know that solving NPO problems in general is at least $\QMA$-hard
\cite{kitaevClassicalQuantumComputation2002},
and even their commutative counterparts are known to be $\NP$-hard \cite{laurentSumsSquaresMoment2009}.
Despite this hardness, a general technique, known as the \emph{sum-of-Hermitian-squares} (SOHS) hierarchy \cite{heltonPositivstellensatzNoncommutativePolynomials2004,pironioConvergentRelaxationsPolynomial2010}, has been developed to systematically obtain a non-decreasing sequence of lower bounds to the optimal value of Problem \eqref{eq:intro_npo}. These lower bounds can also be shown to converge under mild conditions. The SOHS hierarchy is in fact a non-commutative extension of the \emph{sum-of-squares} (SOS) hierarchy for commutative polynomial optimization \cite{lasserreGlobalOptimizationPolynomials2001,lasserreExplicitExactSDP2001,parriloStructuredSemidefinitePrograms2000}. Each level of the SOHS hierarchy consists of a \emph{semidefinite programming} (SDP) relaxation of Problem \eqref{eq:intro_npo}. Although the size of the SDPs grows rapidly as the level increases, the technique has nevertheless found many applications to problems in quantum theory \cite{tavakoliSemidefiniteProgrammingRelaxations2024} beyond the ground state energy problem.

For the NPO formulation of the ground state energy problem, the SOHS hierarchy has been exploited to obtain lower bounds on the ground state energy of many-body quantum systems \cite{nakataVariationalCalculationsFermion2001, 
pironioConvergentRelaxationsPolynomial2010,
baumgratzLowerBoundsGround2012, barthelSolvingCondensedMatterGroundState2012,hanQuantumManybodyBootstrap2020, mazziottiQuantumManyBodyTheory2023, wangCertifyingGroundStateProperties2024, wangScalableGroundStateCertification2026}. Additionally, when designing approximation algorithms for the quantum max-cut problem, solving low levels of the SOHS hierarchy for the ground state energy problem is often used as a subroutine \cite{brandaoProductstateApproximationsQuantum2013,
gharibianAlmostOptimalClassical2019,
parekhApplicationLevel2Quantum2021
}.

For an $n$-qubit Hamiltonian 
$H$,
let 
$d \in \natural$ 
be the maximum Pauli weight of the Pauli strings appearing in the expansion of 
$H$
in the Pauli basis.
For the NPO formulation of the ground state energy of 
$H$
and for
any positive integer
$\ceil{d/2} \le r \le n$, 
the
$r$-th level of the
SOHS hierarchy can be formulated as follows 
\begin{equation}
    \label{eq:intro-sohs}
\begin{aligned}
    H_r  \eqdef \sup & \quad \lambda \\
    \mathrm{s.t.}& \quad H - \lambda \I = \sum_{i} A_i^\dagger A_i \\
    & \quad A_i \in 
    \Span
    (\Pn[\le r])
    \qquad \forall i,
\end{aligned}
\end{equation}
where 
$\Pn[\le r]$ denotes the set of Pauli strings of length $n$ with Pauli weight at most $r$. In effect, by constructing an SOHS polynomial $\sum_i A_i^\dagger A_i$ built from Pauli strings of degree no larger than $r$, we are guaranteeing the positivity of $H- \lambda \I$ (and consequently that $H_{\min} \geq \lambda$). The critical observation of the SOHS hierarchy is that the construction of an SOHS polynomial from Pauli operators of weight at most $r$ can be equivalently formulated as the positive semidefiniteness of some $|\Pn[\le r] | \times |\Pn[\le r] |$ matrix, which turns the problem of computing $H_r$ into an SDP. The size of the SDP grows with $r$ and one has that 
$H_r \le H_{r+1} \le H_{\min}$ for all 
$r \le n$. 
Moreover, the sequence 
$(H_r)_{r \in [n]}$
converges to
$H_{\min}$
in finitely many steps. In particular, for $r=n$,
$H_n = H_{\min}$.

Although the hierarchy converges at a finite level, no rate of convergence for 
$(H_r)_{r \in [n]}$
is known outside this regime. In fact, no rate of convergence results exist for the SOHS hierarchy for any NPO problem.
By contrast, for the SOS hierarchy in commutative polynomial optimization, convergence rate bounds are known in several cases \cite{
nieComplexityPutinarsPositivstellensatz2007, 
kordaConvergenceRatesMomentsumofsquares2017,
kordaConvergenceRatesMomentsumofsquares2018,
fangSumofsquaresHierarchySphere2021,
slotSumofSquaresHierarchiesPolynomial2022,
laurentEffectiveVersionSchmudgens2023,
baldiEffectivePutinarsPositivstellensatz2023,
baldiDegreeBoundsPutinars2024,
schlosserConvergenceRatesMomentSoS2024,
schlosserSpecializedEffectivePositivstellensatze2024,
baldiLojasiewiczInequalitiesEffective2025,
kordaConvergenceRatesSumsofsquares2025,
magronConvergenceRatesPolynomial2025}, including polynomial optimization over the binary cube \cite{slotSumofsquaresHierarchiesBinary2023}, which is equivalent to the ground state energy problem for classical spin systems. We refer the reader to \cite{laurentOverviewConvergenceRates2026} for a review of several existing SOS convergence rate results. 

\subsection{Results}
\label{sec:results}
 In this work, we prove a convergence rate for the sequence 
 $(H_r)_{r \in [n]}$ given by the SOHS hierarchy \eqref{eq:intro-sohs}. This result constitutes, to the best of our knowledge, the first convergence rate of the SOHS hierarchy for any family of NPO problems. The following theorem is our main result.

 \begin{theorem}
    \label{thm:main}
    For fixed $d \leq n$,
    let 
    $H \in \Span(\Pn[\le d])$ 
    be an 
    $n$-qubit Hamiltonian, whose ground state energy is denoted by 
    $H_{\min}$.
    For any positive integer 
    $r < n$,
    let 
    $H_r$ be the bound given by the SOHS hierarchy, defined in \Cref{eq:intro-sohs}, and let 
    $\zeta_{n,r}$ denote the smallest root of the Krawtchouk polynomial $\Kc_{r,4}^n$ defined in \Cref{eq:krawtchouk_definition}. Then, if 
    \begin{equation*}
         d(d+1) (\zeta_{n,r+1} / n )\leq 1/2,
    \end{equation*}
    we have
    \begin{equation}
    \label{eq:main_thm_bound}
        H_{\min} - H_r \leq 2\gamma_d d(d+1) \norm{H} \left( \frac{\zeta_{n, r+1}}{n} \right),
    \end{equation}
    where 
     $\gamma_d \le (1 + \sqrt{2})^d$
    is a constant depending only on 
    $d$ 
    and
    $\|H\|$ is the operator norm of $H$. 
\end{theorem}
It is known from the theory of orthogonal polynomials that for a fixed $n$,  
the sequence 
$(\zeta_{n,r})_{r \in \inter{0}{n}}$
is decreasing \cite[Theorem~3.3.2]{szegoOrthogonalPolynomials1939}, and therefore, our bound in  \Cref{eq:main_thm_bound} is decreasing in 
$r$.
This bound can be made more explicit when we consider a particular asymptotic regime, where 
$n \to \infty$ and $r \approx tn$
for 
$t \in [0,3/4]$. 
In this case,
it is known \cite[Equation~128]{levenshteinKrawtchoukPolynomialsUniversal1995} that 
\[
\lim_{r/n \to t} \frac{\zeta_{n, r}}{n} = 
\varphi(t)
\eqdef
\frac{3 - 2t - 2\sqrt{3t(1-t)}}{4}.
\]
Here, by 
$\lim_{r/n \to t}$,
one means that the convergence holds for any sequence 
$(r_i)_{i \in \natural}$
and
$(n_i)_{i \in \natural}$
such that
\[
n_i \xrightarrow{i\to\infty} \infty
\quad \text{and}
\quad
r_i/n_i \xrightarrow{i\to\infty} t.
\]

Fix $d \in \natural$.
To make the dependence on the number of qubits explicit, for 
$n \geq d$,
let 
$H^{(n)} \in \Pn[\le d] $
denote an 
$n$-qubit 
Hamiltonian,
and let 
$H^{(n)}_{\min}$ 
and  
$H^{(n)}_{r}$
denote, respectively, the ground state energy and the estimate obtained from the 
$r$th level of the SOHS hierarchy.
Combining the above-mentioned asymptotic behavior of the extremal roots of the Krawtchouk polynomials with 
\Cref{thm:main},
we see that if 
\[d(d+1) \varphi(t) < 1/2 ,\]
then
\[
\limsup_{r/n \to t} \frac{H_{\min}^{(n)} - H_r^{(n)}}{\norm{H^{(n)}}} \le  2\gamma_d d(d+1) \varphi(t), 
\]
where 
$\gamma_d$
is the same constant as in 
\Cref{eq:main_thm_bound}.  
From this, we see that for every 
$\epsilon > 0$,
one can choose
$t_{\epsilon} \in [0,3/4]$
such that for all sufficiently large 
$n$ and for
$r \approx t_{\epsilon} n$,
\[
\frac{H_{\min}^{(n)} - H_{r}^{(n)}}{\norm{H^{(n)}}} \le \epsilon.
\]
In fact, as 
$\varphi$ is decreasing and self-inverse, it suffices to take $t_{\epsilon} > \varphi(\frac{\epsilon}{2\gamma_d d(d+1)})$.

\subsection{Outline of the proof technique}
\label{sec:proof_outline}
 We prove \Cref{thm:main} by extending
 the polynomial kernel method~\cite{fangSumofsquaresHierarchySphere2021},
 which was introduced to prove convergence rates of the SOS hierarchy for commutative polynomial optimization problems, to a noncommutative setting. Our proof technique follows a similar structure to the proof of the SOS convergence bounds in \cite{slotSumofsquaresHierarchiesBinary2023} 
 for polynomial optimization over the 
$q$-ary hypercube. 
We refer the interested reader to \cite{laurentOverviewConvergenceRates2026} for an account of this method in the commutative setting and we now briefly overview our extension.

Recall that level $r$ of the SOHS hierarchy asks for 
\begin{equation*}
\begin{aligned}
    H_r \eqdef \sup & \quad \lambda \\
    \mathrm{s.t.}& \quad H - \lambda \I = \sum_{i} A_i^\dagger A_i \\
    & \quad A_i \in 
    \Span
    (\Pn[\le r])
    \qquad \forall i.
\end{aligned}
\end{equation*}
Assume that for some 
$\epsilon(r)$, 
one can show that
$H - (H_{\min} - \epsilon(r))\I$
can be written as an SOHS of degree at most $2r$, i.e., there exist
$t \in \natural$
and 
$S_i \in 
\Span
    (\Pn[\le r] )
$
for 
$i \in [t]$, such that
\begin{equation}
    \label{eq:kernel_method_assumption}
H - \left(H_{\min}-\epsilon(r)\right) \I = \sum_{i=1}^t S_i^{\dagger} S_i.
\end{equation}
Then from the definition of 
$H_r$,
one has
\begin{equation*}
    H_r \ge H_{\min}-\epsilon(r),
\end{equation*}
or equivalently, 
\begin{equation*}
    \label{eq:kernel_method_conclusion}
    H_{\min} - H_r \leq \epsilon(r).
\end{equation*}
Therefore, to show that a function 
$\epsilon(r)$
upper bounds the error of level $r$ of the hierarchy,
it is enough to show that \Cref{eq:kernel_method_assumption}
holds for 
$\epsilon(r)$. 

To this end, assume that there exists a unital linear map
$\Phi: \bounded(\complex^{2^n}) \to \bounded(\complex^{2^n})$,
whose inverse 
$\Phi^{-1}$
is well-defined on 
$\Span \big( \supp(H) \cup \{ \I \}  \big)$,
where 
$\supp(H)$
denotes the set of Pauli strings appearing in the expansion of 
$H$
in the Pauli basis.
Moreover, assume that 
$\Phi^{-1}(H)$
is Hermitian and 
is 
$\epsilon(r)$-close to 
$H$ with respect to the operator norm, i.e.
\[\|\Phi^{-1}(H) - H\| \le \epsilon(r).\]
From this, it follows that
\[
\Phi^{-1}\big( H - (H_{\min} - \epsilon(r))\I \big)
= 
\big( \Phi^{-1}(H) - H + \epsilon(r) \I \big) + \big( H - H_{\min} \I \big) \succeq 0.
\]

Now, assume further that the linear map 
$\Phi$
has the extra property that it maps positive semidefinite operators to SOHS of degree at most 
$2r$, i.e., for any
$X \in \bounded(\complex^{2^n})$ such that 
$X \succeq 0$,
there exist 
$m \in \natural$
and 
$A_1,\dots,A_m \in \Span(\Pn[\le r])$ 
such that 
$\Phi(X) = \sum_{i=1}^m A_i^\dagger A_i$.
Then, by the above line, we have 
\[
H - (H_{\min} - \epsilon(r))\I = \Phi\big(
\Phi^{-1}\big( H - (H_{\min} - \epsilon(r))\I \big)
\big) = \sum_{i=1}^m A_i^\dagger A_i
\]
where the first equality follows from invertibility and the second equality follows from the SOHS property of $\Phi$. This shows that to prove
$
H_{\min} - H_r \le \epsilon(r)
$,
it suffices to find a linear map 
$\Phi$
that satisfies the assumptions discussed above.  

The main technical work of the proof is therefore the construction of such a linear map $\Phi$.  In Section~\ref{sec:kernel-const} we consider maps of the form
\begin{equation}
\label{eq:map_form_outline}
    \Phi_{\bbeta} = \sum_{k=0}^{n} \eta_k \Ec_k,
\end{equation}
where  
$\bbeta = (\eta_0,\dots, \eta_n) \in \real^{n+1}$
and
$\Ec_k: \bounded(\complex^{2^n}) \to \bounded(\complex^{2^n})$
is defined as
\[
\Ec_k(X) \eqdef 2^{-n} \sum_{\Pb \in \Pn[k]} \tr(\Pb X) \Pb,
\]
with $\Pn[k]$ denoting the set of Pauli strings of weight $k$. 
We show that for any map
$\Phi_{\bbeta}$
of the form 
\eqref{eq:map_form_outline},
its corresponding Choi operator
$\choi{\Phi_{\bbeta}}$
admits the closed form
\begin{equation*}
    \label{eq:closed_form_intro}
    \choi{\Phi_{\bbeta}} = 2^{-n} q(D),
\end{equation*}
where 
$q \in \real[x]$
is a univariate polynomial related to 
$\bbeta = (\eta_0,\dots, \eta_{n})$ through certain Krawtchouk polynomials~\cite{Krawtchouk1929},
and 
$D \in \bounded(\complex^{2^n} \otimes \complex^{2^n})$
is a specific Hermitian operator (see \Cref{eq:D_definition}). This characterization allows us to derive sufficient conditions on the polynomial $q$,
ensuring that the associated linear map 
$\Phi$
satisfies the above requirements. In Section~\ref{sec:analysis_of_the_kernel}, we apply the results of 
\cite{slotSumofsquaresHierarchiesBinary2023} to
show the existence of such a polynomial $q$, yielding a linear map satisfying the required conditions with 
\[\epsilon(r) = 2\gamma_d d(d+1) \norm{H} \frac{\zeta_{n, r+1}}{n}.\]

\section{Preliminaries}
\subsection{Notations}
We denote the set of non-negative integers by 
$\natural_{0}
\eqdef \{0\} \cup \natural$.
The imaginary unit is denoted by the boldface letter $\ib$. 
For 
$a \le b \in \natural_{0}$, 
we write 
$\inter{a}{b}$ 
for the set 
$\{a,a+1,\ldots,b\}$, 
and 
$[n] \eqdef \inter{1}{n}$ 
for 
$n \in \natural$.
For a set 
$S$, 
the cardinality of 
$S$ 
is denoted by 
$\lvert S\rvert$. 
For 
$s \in \real$ 
and 
$r \in \mathbb{Z}$, 
the binomial coefficient 
$\binom{s}{r}$ 
is defined as
\begin{equation*}
    \binom{s}{r} \eqdef 
        \begin{cases}
        \frac{s(s-1)\cdots (s-r+1)}{r!} & \text{if $r \in \natural$}, \\
        1 & \text{if $r=0$}, \\
        0 & \text{otherwise}.	
        \end{cases}
\end{equation*}
We denote by 
$\delta_{x,y}$ the Kronecker delta, which is equal to $1$ when 
$x$ and $y$ are equal, and is $0$ otherwise.

For a vector space 
$V$ 
and a subset 
$S \subseteq V$, 
we write 
$\Span(S)$ 
for the linear span of the elements in
$S$.   
For a Hilbert space 
$\Hc$,  
the algebra of bounded operators  
$X: \Hc \to \Hc$ 
is denoted by 
$\bounded(\Hc)$, and the real vector space of 
bounded self-adjoint operators on $\Hc$ is denoted by
$\bounded^{\mathrm{sa}}(\Hc)$.
The operator norm of 
$X \in \bounded(\Hc)$ 
is denoted by 
$\norm{X}$. 
We denote the partial trace of an operator 
$X \in \bounded(\Hc_1 \otimes \Hc_2)$ 
over 
$\Hc_1$ by $\tr_1(X) \in \bounded(\Hc_2)$. 
An element 
$X \in \bounded(\Hc)$ 
is called \emph{positive semidefinite}, denoted by 
$X \succeq 0$, 
if 
$X = Y^\dagger Y$ 
for some 
$Y \in \bounded(\Hc)$. 
For 
$X, Y \in \bounded(\Hc)$, 
we write 
$X \succeq Y$ 
if 
$X - Y \succeq 0$.

Single-qubit Pauli operators are denoted by 
$\Iq, \X, \Y, \Z$.
We denote the set of 
$n$-qubit 
Pauli strings by 
$\Pn \eqdef 
\{
    \mathrm{P}_{\!1}\otimes\cdots\otimes \mathrm{P}_{\!n} : 
    \mathrm{P}_i
    \in
    \{\Iq,\X,\Y,\Z\}
\}$ and we denote elements of this set by boldface uppercase letters, e.g., 
$\Pb = \mathrm{P}_{\!1}\otimes\cdots\otimes \mathrm{P}_{\!n}$. 
The identity operator on 
$\complex^{2^n}$ 
is denoted by 
$\I = \Iq^{\otimes n}$.  
For 
$\Pb \in \Pn$, 
the \emph{weight} of 
$\Pb$, 
denoted by 
$\wt(\Pb)$, 
is defined as the number of indices 
$i \in [n]$ 
such that 
$\mathrm{P}_{\!i} \neq \Iq$. 
For 
$k \in \inter{0}{n}$, 
we write 
$\Pn[k]$ 
and 
$\Pn[\le k]$ 
for the subsets of 
$\Pn$ 
consisting of Pauli strings of weight exactly $k$ and at most $k$, respectively. 

The identity map 
$\bounded(\complex^{2^n}) \to \bounded(\complex^{2^n})$ 
is denoted by 
$\id_n$, 
and we often omit the subscript 
$n$ 
when it is clear from the context. 
For a linear map 
$\Phi: \bounded(\complex^{2^n}) \to \bounded(\complex^{2^n})$, 
we denote by $\choi{\Phi}$
the \emph{Choi matrix} of 
$\Phi$ 
with respect to the computational basis
$\{\ket{i}\}_{i \in \{0,1\}^n}$, defined as 
\begin{equation}
    \label{eq:choi_definition}
        \choi{\Phi} \eqdef \sum_{i,j \in \{0,1\}^n} \ketbra{i}{j} \otimes \Phi(\ketbra{i}{j}) \in \bounded(\complex^{2^n} \otimes \complex^{2^n}).
\end{equation}
One can recover the action of the map 
$\Phi$
from its Choi matrix, as
\begin{equation}
\label{eq:choi_inversion}
    \Phi(X) = \tr_1 \bigl[
    \choi{\Phi} (X^{\top} \otimes \I)
    \bigr]
\end{equation}
for all 
$X \in \bounded(\complex^{2^n})$, where $X^{\top}$ denotes the transpose of $X$.
For any linear operator 
$X \in \Bc(\complex^{2^n})$, 
we define the \emph{Choi state of 
$X$}, 
denoted by 
$\vec{X}$,
as
\begin{equation}
    \label{eq:vectorization_definition}
    \vec{X} 
    \eqdef 
    \sum_{i\in \{0,1\}^n} 
    \ket{i} 
    \otimes 
    X \ket{i} 
    \in 
    \complex^{2^n} \otimes \complex^{2^n}.
\end{equation}
We denote 
$\vec{X}^\dagger$ 
by 
$\cev{X}$, 
and the inner product 
$\cev{X}\vec{Y}$ 
is denoted by 
$\cevvec{X}{Y}$.

For 
$2 \le q \in \natural$,
$n \in \natural$ 
and
$k \in \natural_0$, 
the \emph{$q$-ary Krawtchouk polynomial of degree 
$k$ with parameter $n$} is defined as 
\begin{equation}
    \label{eq:krawtchouk_definition}
    \mathcal{K}_{k,q}^n(x) := \sum_{i=0}^k (-1)^i (q-1)^{k-i} \binom{x}{i} \binom{n-x}{k-i}.
\end{equation}
In \Cref{tab:krawtchouk_identities}, we list three identities for these polynomials, which will be used later. The reader can find proofs of these identities in \cite{levenshteinKrawtchoukPolynomialsUniversal1995}.
In particular, from \Cref{eq:krawtchouk_orthogonality} in \Cref{tab:krawtchouk_identities}, one can see that
for any 
$n \in \natural$,
Krawtchouk polynomials 
$\{\Kc_{k,q}^n\}_{k \in \inter{0}{n}}$ 
form an orthogonal basis for the space of real univariate polynomials of degree at most $n$,
with respect to the discrete measure
\[
\mu(x) = q^{-n} \sum_{t=0}^n (q-1)^{t} \binom{n}{t} \delta_{x,t},
\]
on 
$\inter{0}{n}$.
Furthermore, note that by the Lagrange interpolation theorem, any function $f: \inter{0}{n} \to \real$, over 
$\inter{0}{n}$, is equal to a unique polynomial of degree at most 
$n$. Therefore, for any 
$p \in \real[x]$,
there exist
$(\alpha_0, \alpha_1, \dots, \alpha_{n}) \in \real^{n+1}$,
such that
\[
p(x) = \sum_{k = 0}^n \alpha_k \Kc_{k,q}^n (x), \quad \text{for } x \in \inter{0}{n}.
\]

\begin{table}[h]
\centering
\renewcommand{\arraystretch}{0.4}

\begin{tabular}{|>{\raggedright\arraybackslash}m{0.7\textwidth}|}
\hline

\fbox{\textbf{Orthogonality relation I}}\quad
\emph{For
$2 \le q \in \natural$, $n \in \natural$, and $k, k' \in \inter{0}{n}$,}
\begin{equation}
\label{eq:krawtchouk_orthogonality}
    \sum_{x=0}^n 
    \mathcal{K}_{k,q}^n(x)\,
    \mathcal{K}_{k',q}^n(x)\,
    q^{-n} (q-1)^x \binom{n}{x}
    = (q-1)^k \binom{n}{k} \delta_{k,k'}.
\end{equation}
\\ \hline

\fbox{\textbf{Orthogonality relation II}}\quad
\emph{For
$2 \le q \in \natural$, $n \in \natural$, and $k', k'' \in \inter{0}{n}$,}
\begin{equation}
    \label{eq:krawtchouk_orthogonality_implied}
    \sum_{k=0}^n \Kc_{k',q}^n(k) \Kc_{k,q}^n(k'') = q^n \delta_{k',k''}.
\end{equation}
\\ \hline

\fbox{\textbf{A special case of \Cref{eq:krawtchouk_orthogonality_implied}}}\quad
\emph{For
$2 \le q \in \natural$, $n \in \natural$, $x \in \inter{0}{n}$,}
\begin{equation} \label{eq:krawtchouk_orthogonality_special}
    \sum_{k=0}^n \Kc_{k,q}^n(x) 
    = q^n \delta_{0,x}.
\end{equation}
\\ \hline

\end{tabular}

\caption{Identities for Krawtchouk polynomials used in this work.}
\label{tab:krawtchouk_identities}
\end{table}

\subsection{Noncommutative polynomial optimization}
In this section, we explain how the ground state energy problem can be formulated as a non-commutative polynomial optimization (NPO) problem, and how the sum-of-Hermitian-squares (SOHS) hierarchy for this formulation provides a sequence of lower bounds on the ground state energy. In this part, we specialize our review of the NPO framework and the SOHS hierarchy to the setting of the Pauli algebra, as this is the setting relevant to our work. The interested reader is referred to \cite{pironioConvergentRelaxationsPolynomial2010} for a general introduction to the NPO framework and the SOHS hierarchy.
\subsubsection{Noncommutative polynomial optimization over the Pauli algebra}
Consider the set of non-commuting variables
$\Xc_n \eqdef \bigcup_{i=1}^n \{\pauli{x}{i}, \pauli{y}{i}, \pauli{z}{i}\}$.
Let 
$\langle \Xc_n \rangle$
denote the set of all finite words formed from the letters in 
$\Xc_n$ 
by concatenation, with the empty word denoted by 
$1$. 
We denote the set of all complex polynomials in these non-commuting variables by
$\mathbb{C}\langle \Xc_n \rangle$.
The \emph{degree} of a monomial in 
$\mathbb{C}\langle \Xc_n \rangle$
is defined as the length of the corresponding word, and the degree of a polynomial is defined as the maximum degree of monomials appearing in the polynomial with non-zero coefficients.
The ring 
$\mathbb{C}\langle \Xc_n \rangle$ 
is equipped with an involution 
$p \mapsto p^*$, 
that is a conjugate-linear map fixing the letters in 
$\Xc_n$ 
and reversing the order of words, making 
$\mathbb{C}\langle \Xc_n \rangle$ 
a unital 
$*$-algebra.
A polynomial 
$p \in \mathbb{C}\langle \Xc_n \rangle$ 
is called \emph{Hermitian} if 
$p^* = p$. 

In this paper, we consider the Pauli algebra 
$\Pcc$, 
which is the algebra of complex non-commutative polynomials in variables 
$\Xc_n$ subject to the equalities\footnote{
Formally, 
$\Pcc = \mathbb{C}\langle \Xc_n \rangle / I_n$, 
where 
$I_n$ 
is the two-sided ideal generated by the polynomials corresponding to the equalities in \Cref{eq:pauli_algebra_relations}.
}
\begin{equation}
    \label{eq:pauli_algebra_relations}
    \begin{aligned}
        &(\pauli{x}{i})^2 = (\pauli{y}{i})^2 = (\pauli{z}{i})^2 = 1, \quad \text{for } i \in [n], \\
        &\pauli{x}{i}\pauli{y}{i} - \ib \pauli{z}{i} =  
        \pauli{y}{i}\pauli{z}{i} - \ib \pauli{x}{i} =
        \pauli{z}{i}\pauli{x}{i} - \ib \pauli{y}{i} = 0, \quad \text{for } i \in [n], \\
        & \pauli{A}{i} \pauli{B}{j} - \pauli{B}{j} \pauli{A}{i} = 0, \quad \text{for } i,j \in [n],\ i\neq j,\ \mathsf{A}, \mathsf{B} \in \{\mathsf{X},\mathsf{Y},\mathsf{Z}\}.
    \end{aligned}
\end{equation}
Let $f \in \Pcc$ be a Hermitian polynomial in the Pauli algebra.
The non-commutative polynomial optimization problem with objective $f$ over the Pauli algebra $\Pcc$ is the optimization problem defined as
\begin{equation}
    \label{eq:NPO_general}
    \begin{aligned}
        f_{\min} 
        \eqdef
        \inf_{
                \tau, \ket{\psi}
        }
        \bra{\psi} \tau(f) \ket{\psi},
    \end{aligned}
\end{equation}
where the infimum is taken over all $*$-representations
$\tau: \Pcc \to \bounded(\Hc)$
on arbitrary separable Hilbert spaces $\Hc$, and all unit vectors $\ket{\psi} \in \Hc$.
This optimization problem can be equivalently expressed as
\begin{equation}
\label{eq:NPO_min_eigenvalue}
\begin{aligned}
f_{\min} 
= \sup \left\{ \lambda \in \real : \;\right.
& \tau(f) - \lambda \I_\Hc \succeq 0, \text{for all $*$-representations }\tau: \Pcc \to \bounded(\Hc)\\
& \text{ on arbitrary separable Hilbert spaces $\Hc$} \left.\right\}.
\end{aligned}
\end{equation}
The \emph{sum-of-Hermitian-squares (SOHS) hierarchy} for the NPO problem in \Cref{eq:NPO_general} is a family of relaxations obtained from replacing the condition that 
$\tau(f) - \lambda \I_\Hc \succeq 0$ 
for all 
$*$-representations 
$\tau$ by the stronger condition that 
$f - \lambda 1$ can be expressed as a sum of Hermitian squares of a certain degree on 
$\Pcc$. 
For 
$r \in \natural$, 
the 
$r$-th level of the SOHS hierarchy is defined as
\begin{equation}
    \label{eq:SOHS_hierarchy}
    \begin{aligned}
      f_{r} \eqdef \sup \{ \lambda \in \real : & f - \lambda 1 = \sum_i g_i^* g_i
    \text{ modulo the equalities in \eqref{eq:pauli_algebra_relations},} \\
    &\text{where }
    g_i \in \complex \langle \Xc_n \rangle   \text{ and } \deg(g_i) \le r \text{ for all $i$} \}.
    \end{aligned}
\end{equation}
Clearly, one has $f_{r} \le f_{r+1} \le f_{\min}$ for all $r \in \natural$. Moreover, $f_r$ can be computed by solving a semidefinite program. 
\subsubsection{Ground state energy as an NPO problem}
To write the ground state problem \eqref{eq:gs_intro} as an NPO problem, we use the following well-known lemma. For completeness, we provide a proof of this lemma in Appendix~\ref{appndx:pauli_algebra_representation}.
\begin{lemma}
    \label{lem:Pauli-algebra-representation}
    The Pauli algebra 
    $\Pcc$
    admits a unique (up to unitary equivalence) irreducible $*$-representation  
    $\pi: \Pcc \to \bounded(\complex^{2^n})$ 
    given by
    \begin{equation}
    \label{eq:representation}
    \begin{aligned}
    \pi(\pauli{x}{i}) &= \Iq^{\otimes (i-1)} \otimes \X \otimes \Iq^{\otimes (n-i)}, \\
    \pi(\pauli{y}{i}) &= \Iq^{\otimes (i-1)} \otimes \Y \otimes \Iq^{\otimes (n-i)}, \\
    \pi(\pauli{z}{i}) &= \Iq^{\otimes (i-1)} \otimes \Z \otimes \Iq^{\otimes (n-i)}.
    \end{aligned}
    \end{equation}
\end{lemma} 
\noindent Given an 
$n$-qubit 
Hamiltonian 
$H \in \bounded(\complex^{2^n})$, 
we can express it in the Pauli basis as
\begin{equation*}
    H = \sum_{\Pb \in \Pn} h_{\Pb} \Pb,
\end{equation*}
for some coefficients $h_{\Pb} \in \complex$. 
Using the representation $\pi$ in \Cref{lem:Pauli-algebra-representation}, we can identify $H$ with the Hermitian polynomial 
\begin{equation*}
    f^H \eqdef \sum_{\Pb \in \Pn} h_{\Pb} \pi^{-1}(\Pb) \in \Pcc,
\end{equation*}
where 
\begin{equation*}
    \pi^{-1}(\Pb) = \pi^{-1}(\mathrm{P}_{\!1}) \cdots \pi^{-1}(\mathrm{P}_{\!n}) \in \Pcc, \quad \text{for } \Pb = \mathrm{P}_{\!1} \otimes \cdots \otimes \mathrm{P}_{\!n} \in \Pn.
\end{equation*}
It is then a consequence of \Cref{lem:Pauli-algebra-representation} that minimizing $f^H$ over the Pauli algebra 
$\Pcc$, 
as in \Cref{eq:NPO_min_eigenvalue}, is equivalent to finding the ground state energy of 
$H$:
\begin{equation}
    \label{eq:ground_state_NPO}
    f^H_{\min} = H_{\min} \eqdef \sup \{ \lambda \in \real : H - \lambda \I \succeq 0 \}.
\end{equation}
Similarly, the $r$-th level of the SOHS hierarchy
\eqref{eq:SOHS_hierarchy}
for $f^H$ can be expressed as
\begin{equation}
    \label{eq:SOHS_hierarchy_Hamiltonian}
    f^H_{r} = H_r \eqdef \sup \{ \lambda \in \real : H - \lambda \I = \sum_i A_i^\dagger A_i, \text{ where } A_i \in \Span(\Pn[\le r]) \}.
\end{equation}
The SOHS hierarchy \eqref{eq:SOHS_hierarchy_Hamiltonian} for the ground state energy problem \eqref{eq:ground_state_NPO} has finite convergence, i.e. there exists 
$r \in \natural$ 
such that 
$f^H_r = f^H_{\min}$. 
In fact, as
$H - f^H_{\min} \I \succeq 0$,
there exists an operator 
$A \in \bounded(\complex^{2^n})$
such that
$H - f^H_{\min} \I = A^\dagger A$,
and 
$A$ 
can be expressed as a linear combination of Pauli strings of weight at most 
$n$, 
which implies that 
$f^H_n = f^H_{\min}$ 
for 
$r = n$.

\section{Construction of the linear map}
\label{sec:kernel-const}

In \Cref{sec:proof_outline}, we explained how the existence of a linear map 
$\Phi: \bounded(\complex^{2^n}) \to \bounded(\complex^{2^n})$, that satisfies certain conditions, results in an upper bound on the error of the level $r$ SOHS hierarchy relaxation. The following statement, which we have already proved there, summarizes that discussion.
\begin{prop}
    \label{prop:kernel_method}
	Let $H \in \Span(\Pn[\le d])$ be an $n$-qubit Hamiltonian. For 
    $r \le n$,
    assume that there exists a linear map 
    $\Phi:  
     \bounded(\complex^{2^n}) \to \bounded(\complex^{2^n})
    $
    such that
	\begin{enumerate}
        \setcounter{propitem}{-1}
        \pitem\label{P0}
        $\Phi(\I) = \I$.
		\pitem\label{P1} 
        $\Phi^{-1}$
        is well-defined on 
        $\Span \big( \supp(H) \cup \{ \I \}  \big)$,
        where 
        $\supp(H)$
        denotes the set of Pauli strings appearing in the expansion of 
        $H$
        in the Pauli basis. Moreover, 
        $\Phi^{-1}(H)$ is Hermitian. 
        \pitem\label{P2}
        $\|\Phi^{-1}(H) - H\| \le \epsilon(r)$ for some $\epsilon(r) \geq 0$.
		\pitem\label{P3}
        For all 
        $X \in  \bounded(\complex^{2^n})$ such that 
        $X \succeq 0$,
		there exist 
        $m \in \natural$
        and 
        $A_1,\dots,A_m \in \Span(\Pn[\le r])$ 
        such that 
        $\Phi(X) = \sum_{i=1}^m A_i^\dagger A_i$.
	\end{enumerate}
	Then, there exists 
    $t \in \natural$
    and
    $S_1, \dots, S_t \in \Span(\Pn[\le r])$
    such that
	\[H - (H_{\min} - \epsilon(r))\I = \sum_{i \in [t]} S_i^\dagger S_i,\] 
    and hence,
    \[H_{\min} - H_r \le \epsilon(r).\]
\end{prop}

Our aim is now to find a linear map 
$\Phi$
satisfying the conditions 
\Prop{P0}-\Prop{P3} in \Cref{prop:kernel_method}. We will restrict this search to linear maps of the form 
\begin{equation}
\label{eq:restricted_kernel}
    \Phi_{\bbeta} \eqdef \sum_{k=0}^{n} \eta_k \Ec_k,
\end{equation}
where 
$\bbeta= (\eta_0,\dots, \eta_{n}) \in \real^{n+1}$
and
$\Ec_k: \bounded(\complex^{2^n}) \to \bounded(\complex^{2^n})$
is defined
for 
$k \in \inter{0}{n}$
as
\begin{equation}
\label{eq:EC_def}
    \Ec_k(X) \eqdef 2^{-n} \sum_{\Pb \in \Pn[k]} \tr(\Pb X) \Pb.
\end{equation}
From 
\Cref{eq:EC_def}, we see that
for any 
$k \in \inter{0}{n}$,
$\Ec_k$ is the orthogonal projection, with respect to the Hilbert-Schmidt inner-product, onto 
$\Span(\Pn[k])$.
Therefore,
for all
$k \in \inter{0}{n}$, we have that
$\eta_k$
is an eigenvalue of the map 
$\Phi_{\bbeta}$
with the corresponding eigenspace 
$\Span(\Pn[k])$.
From this observation, we obtain conditions on $\eta_0, \dots, \eta_{n}$ under which $\Prop{P0}$, $\Prop{P1}$, and $\Prop{P2}$ are satisfied. 

First note that by imposing 
$\eta_0 = 1$, we have $\Phi_{\bbeta}(\I) = \I$ and hence
$\Phi_{\bbeta}$
will satisfy \Prop{P0}.
For \Prop{P1}, observe that if 
$\eta_k \neq 0$ 
for all
\[
k \in \{ \wt(\Pb)~:~ \Pb \in \supp(H) \cup \{ \I \} \},
\]
then 
$\Phi_{\bbeta}$ 
is invertible on 
$\Span \big(\supp(H) \cup \{ \I \}  \big)$, 
with the inverse having the action
\begin{equation}\label{eq:phi-inverse}
    \Phi_{\bbeta}^{-1} (\Pb) = \frac{1}{\eta_{\wt(\Pb)}} \Pb, \quad \text{for all }\Pb \in \supp(H) \cup \{ \I \}.
\end{equation} 
From this, it immediately follows that 
$\Phi_{\bbeta}^{-1}(H)$
is Hermitian.\\
\smallskip

For \Prop{P2}, we note that 
$\sum_{k=0}^{n} \Ec_k = \id$ and for $H \in \Span(\Pn[\le d])$ we have $\Ec_k(H) = 0$ if $k > d$. Thus,
\begin{equation} 
    \begin{aligned} 
        \|\Phi^{-1} (H) - H\| & = \|\sum_{k=0}^{d} (\eta_k^{-1}-1) \Ec_k(H)\|\\ & \le \max_{j=0,\dots,d} \| \Ec_j(H)\| \cdot \sum_{k=0}^{d} |\eta_k^{-1}-1|,    
    \end{aligned} 
\end{equation}
where the inequality follows from the triangle inequality. 
As a consequence, any choice of $\epsilon(r)$ such that
\[
\max_{j=0,\dots,d} \| \Ec_j(H)\| \cdot \sum_{k=0}^{d} |\eta_k^{-1}-1| \le \epsilon(r),
\]
results in \Prop{P2} holding.

\subsection{Ensuring property \Prop{P3}}
The final property \Prop{P3} requires some additional work to justify. In the following, we derive an alternative expression for the map 
$\Phi_{\bbeta}$ in terms of the coefficients 
$\bbeta= (\eta_0,\dots,\eta_{n})$, connecting them to Krawtchouk polynomials. This connection enables us to find a sufficient condition on 
$\bbeta = (\eta_0, \dots, \eta_{n})$
to satisfy 
\Prop{P3}.

Recall the definition of the maps
$\{\Ec_k\}_{k \in \inter{0}{n}}$ from \Cref{eq:EC_def},
\[
\Ec_k(X) \eqdef 2^{-n} \sum_{\Pb \in \Pn[k]} \tr(\Pb X) \Pb,
\]
and let 
$\{\Ac_k\}_{k \in \inter{0}{n}}$
be another family of maps
$\bounded(\complex^{2^n}) \to \bounded(\complex^{2^n}),$
whose action on 
$X \in \bounded(\complex^{2^n})$
is defined as
\begin{equation}
    \label{eq:Ac_definition}
\begin{aligned}
\Ac_k(X) & \eqdef \sum_{\Qb \in \Pn[k]} \Qb X \Qb.
\end{aligned}
\end{equation}
As we prove in Appendix~\ref{sec:apndx_proofs_kernel_const} 
(see \Cref{prop:krawtchouk_transform}), for $k \in \inter{0}{n}$, we have
\begin{align}
        \Ac_k & = \sum_{t=0}^n \mathcal K_{k,4}^n(t)\,\Ec_t,  \label{eq:krawtchouk_transform_1}\\
        \Ec_k & = 4^{-n}\sum_{t=0}^n \mathcal K_{k,4}^n(t)\,\Ac_t, \label{eq:krawtchouk_transform_2}
\end{align}
where 
$\mathcal K_{k,4}^n$ 
denotes the 
$4$-ary 
Krawtchouk polynomial of degree 
$k$, defined in \Cref{eq:krawtchouk_definition}. 
Since
$\Phi \mapsto \choi{\Phi}$
is a linear isomorphism,
we have similar transformations between their respective Choi matrices:
\begin{align}
\choi{\Ac_k} & = \sum_{t=0}^n \mathcal K_{k,4}^n(t)\,\choi{\Ec_t}, \label{eq:krawtchouk_transform_choi_1}\\
\choi{\Ec_k} & = 4^{-n}\sum_{t=0}^n \mathcal K_{k,4}^n(t)\,\choi{\Ac_t}. \label{eq:krawtchouk_transform_choi_2}
\end{align}
Writing the definition of the Choi matrix, we can also obtain an explicit form for 
$\choi{\Ac_k}$
and
$\choi{\Ec_k}$ 
(see \Cref{prop:choi_PDMaps} in Appendix~\ref{sec:apndx_proofs_kernel_const}):
\begin{align}
    J_{\Ac_k} & = \sum_{\Pb \in \Pn[k]} \vec{\Pb} \cev{\Pb}, \label{eq:choi_Ac_k_explicit}\\
    J_{\Ec_k} & = 2^{-n} \sum_{\Pb \in \Pn[k]} \Pb^\top \otimes \Pb, \label{eq:choi_Ec_k_explicit}
\end{align}
where 
$\vec{\Pb}$ 
is the Choi state of 
$\Pb$
defined in 
\Cref{eq:vectorization_definition}.
From
$\cevvec{\Pb}{\Qb} = 2^n \delta_{\Pb, \Qb}$ (see \Cref{prop:orthogonal_projection} in Appendix~\ref{sec:apndx_proofs_kernel_const}), it follows that the family of operators 
$\{
    \widehat{J}_{\Ac_k}
\}_{k \in \inter{0}{n}}
$
defined as
\begin{equation}
    \label{eq:normalized_choi_Ac}
    \widehat{J}_{\Ac_k} \eqdef 2^{-n} J_{\Ac_k},
\end{equation}
forms a set of orthogonal projectors that satisfy 
$
    \sum_{k=0}^n \widehat{J}_{\Ac_k} = \I \otimes \I
$.
Consequently, the operator
$D \in \bounded(\complex^{2^n} \otimes \complex^{2^n})$ defined as
\begin{equation}
    \label{eq:D_definition}
    D \eqdef \sum_{t=0}^n t 
    \widehat{J}_{\Ac_t}\,,
\end{equation}
is diagonalizable with eigenvalues 
$t = 0, 1, \dots, n$.
Therefore, for any univariate polynomial
$q \in \real[x]$,
\begin{equation}
    \label{eq:spectral_reminder}
    q(D) = \sum_{t=0}^n q(t) \widehat{J}_{\Ac_{t}}.
\end{equation}
\begin{prop}
    \label{prop:kernel_operator}
    Let $q \in \real[x]$ be a univariate polynomial, where
    \begin{equation}
    \label{eq:q_krawtchouk_expansion_assumption}
        q(x) = \sum_{k=0}^{n} \eta_k \mathcal K_{k,4}^n(x), \quad \text{for all $x \in \inter{0}{n}$},
    \end{equation}
    for $\bbeta = (\eta_0,\ldots,\eta_{n}) \in \real^{n+1}$.
    Then the following hold:
    \begin{enumerate}
    \item         
    $q(D) = \sum_{k=0}^{n} \eta_k \sum_{\Pb \in \Pn[k]} \Pb^\top \otimes \Pb$, 
    \item 
    Denoting the linear map corresponding to the Choi matrix 
    $2^{-n} q(D)$ by 
    $\Phi_{2^{-n}q(D)}$,
    we have
    \begin{equation}
    \label{eq:alternative_form}
    \Phi_{2^{-n}q(D)}
    = \Phi_{\bbeta} 
    = \sum_{k=0}^{n} \eta_k \Ec_k.
    \end{equation}
    \end{enumerate}
\end{prop}
\begin{proof}
    \begin{enumerate}
        \item We use the assumption \eqref{eq:q_krawtchouk_expansion_assumption} to rewrite \Cref{eq:spectral_reminder};
    \begin{equation}
    \label{eq:kernel_in_Ec}
    \begin{aligned}
         q(D) 
            & =
            \sum_{t=0}^n q(t) \widehat{J}_{\Ac_t} \\
            & =
            \sum_{t=0}^n \left(\sum_{k=0}^{n} \eta_k \mathcal K_{k,4}^n(t)\right) \widehat{J}_{\Ac_t} \\
            & =
            2^{-n} \sum_{k=0}^{n}  \eta_k \left(\sum_{t=0}^n \mathcal K_{k,4}^n(t) J_{\Ac_t}\right)\\
            &=2^n\sum_{k=0}^{n} \eta_k J_{\Ec_k},
    \end{aligned}
    \end{equation}
    where the last line
    follows from \Cref{eq:krawtchouk_transform_choi_2}, that is,
    $\sum_{t=0}^n \mathcal K_{k,4}^n(t) J_{\Ac_t} = 4^n J_{\Ec_k}$. 
    From \Cref{eq:choi_Ec_k_explicit}, we have
    $
        2^n \choi{\Ec_k} =
        \sum_{\Pb \in \Pn[k]} \Pb^\top \otimes \Pb
    $, which implies
    \[
    q(D) = \sum_{k=0}^{n} \eta_k \sum_{\Pb \in \Pn[k]} \Pb^\top \otimes \Pb.
    \]
    \item
    For any
    $X \in \bounded(\mathbb{C}^{2^n})$, by the definition of the Choi matrix and \Cref{eq:kernel_in_Ec}, we have 
    \begin{align*}
    \Phi_{2^{-n}q(D)}(X)
    & =
    \Tr_1\!\bigl[2^{-n}q(D)(X^\top \otimes \I)\bigr] \\
    & =
    \sum_{k=0}^{n} \eta_k \Tr_1\!\bigl[\choi{\Ec_k}(X^\top \otimes \I)\bigr].
    \end{align*}
    Recalling from \Cref{eq:choi_inversion} that
    $\Tr_1\!\bigl[\choi{\Ec_k}(X^\top \otimes \I)\bigr] = \Ec_k (X)$, we get
    $\Phi_{2^{-n}q(D)} = \sum_{k=0}^{n} \eta_k \Ec_k$.
\end{enumerate}
\end{proof}

From
\Cref{prop:kernel_operator} 
we see that any univariate polynomial 
$q$ of degree at most 
$2r$
gives rise to a Choi matrix 
\begin{equation}
        C_q \eqdef 2^{-n}q(D),
\end{equation}
and, in turn, to the linear map
$\Phi_{2^{-n} q(D)} = \sum_{k=0}^{n} \eta_k \Ec_k$,
where 
$\bbeta = (\eta_0, \dots, \eta_{n})$
satisfy
\[
q(x) = \sum_{k = 0}^n \eta_k \Kc_{k,4}^n (x), \quad \text{for } x \in \inter{0}{n}.
\] 
Our aim is to show that if we take 
\begin{equation*}
    q(x) = u^2(x),
\end{equation*}
for some 
$u \in \real[x]$ of degree $r$,
then the linear map corresponding to the Choi matrix 
$C_q$
satisfies 
\Prop{P3}. To show this, we  need the following lemma.

\begin{lemma}
	\label{prop:kernel-sos}
	Let 
    $C = R^\dagger R$
    for 
    $R \in \bounded(\complex^{2^n} \otimes \complex^{2^n})$, 
    and let
	$$R = \sum_{\Pb,\Qb \in \Pn} R_{\Pb,\Qb} \Pb \otimes \Qb, \quad \text{for some } R_{\Pb,\Qb} \in \complex.$$
    If 
    $\Qb \in \Pn[\le r]$
    for every pair 
    $(\Pb,\Qb)$
    with
    $R_{\Pb,\Qb} \neq 0$,
	then the linear map corresponding to 
    $C$, denoted by
    $\Phi_C$,
    satisfies 
    \Prop{P3} 
    in \Cref{prop:kernel_method}.
\end{lemma}
\begin{proof}
	Let $X \succeq 0$.
    Then, its transpose 
    $X^\top$ 
    is also positive semidefinite; hence, 
    $\sqrt{X^\top}$ 
    is well-defined.
	Consider the operator
	\begin{equation}
        \label{eq:A_definition}
		A \eqdef
		R (\sqrt{X^\top} \otimes \I)
		= 
		\sum_{\Pb,\Qb \in \Pn} R_{\Pb,\Qb} \Pb \sqrt{X^\top} \otimes \Qb.
	\end{equation}
	For any 
    $\Pb \in \Pn$, 
    let 
    $\Pb \sqrt{X^\top} = \sum_{k,\ell \in \{0,1\}^n} \alpha_{k,\ell}^{\Pb} \ketbra{k}{\ell}$ 
    for some coefficients
    $\alpha_{k,\ell}^{\Pb} \in \complex$. 
    Expanding $A$ we have
	\begin{align*}
		A & = 
        \sum_{\Pb,\Qb \in \Pn} 
        R_{\Pb,\Qb} 
        \Bigl(
            \sum_{k,\ell \in \{0,1\}^n} 
            \alpha_{k,\ell}^{\Pb} 
            \ketbra{k}{\ell}
        \Bigr) \otimes \Qb \\
        & = 
        \sum_{k,\ell \in \{0,1\}^n}
        \ketbra{k}{\ell} \otimes \Bigl(
            \sum_{\Pb,\Qb \in \Pn}
            R_{\Pb,\Qb}
            \alpha_{k,\ell}^{\Pb}
            \Qb
        \Bigr).
	\end{align*}
	Let us define
    \begin{equation}
        A_{k,\ell} \eqdef 
        \sum_{\Pb,\Qb \in \Pn}
        R_{\Pb,\Qb}
        \alpha_{k,\ell}^{\Pb}
        \Qb.
    \end{equation}
    From the assumption on $R_{\Pb,\Qb}$, we have $A_{k,\ell} \in \Span(\Pn[\le r])$ for all $k,\ell \in \{0,1\}^n$.
    To conclude, we note that		
	\begin{align*}
		\Phi_{C}(X) & = 
         \Tr_1\!\bigl[C(X^\top \otimes \I)\bigr] \\
		& = 
        \Tr_1\!
        \bigl[
            (\sqrt{X^\top} \otimes \I) 
            R^\dagger R ( \sqrt{X^\top} \otimes \I)
        \bigr]  \\
		& =
          \Tr_1\!\bigl[A^\dagger A\bigr] \\
        & = \sum_{k,\ell \in \{0,1\}^n} \left(A_{k,\ell}\right)^\dagger 
        \left( A_{k,\ell}\right)\,,
\end{align*}		
showing the output is an SOHS 
polynomial of degree at most $2r$.
\end{proof} 

Recall from 
\Cref{eq:D_definition} that 
$D$
is a Hermitian operator. 
Thus, for any degree-$r$ polynomial
$u \in \real[x]$,
    \begin{equation}
        C_{u^2} = 2^{-n} u^2(D) = 2^{-n} u(D)^\dagger u(D).
    \end{equation}
Since 
    $u$
    is of degree 
    $r \le n$,
    it can be expanded in the Krawtchouk basis 
    $\{\Kc_{k,4}^n\}_{k \in \inter{0}{n}}$, resulting in the polynomial identity
    $u(x) = \sum_{k=0}^{r} u_k \Kc_{k,4}^n (x)$.
Then,
from 
    \Cref{prop:kernel_operator}
    we know that
    \begin{equation}
        u(D) = \sum_{k=0}^r u_k \sum_{\Pb \in \Pn[k]} \Pb^\top \otimes \Pb.
    \end{equation}
This implies that
    $C_{u^2}$
satisfies the premise of 
    \Cref{prop:kernel-sos},
    and consequently, its corresponding linear map
    $\Phi_{2^{-n} u^2(D)}$ 
    satisfies 
    \Prop{P3}.

    The following theorem gives a summary of the sufficient conditions for satisfying \Prop{P0}-\Prop{P3} that we have established.

    \begin{theorem}
    \label{thm:ev_conditions}
        Let
        $H \in \Span(\Pn[\le d])$
        be an 
        $n$-qubit Hamiltonian, and let
        $u \in \real[x]$
        be a univariate polynomial of degree at most 
        $r \le n$, with 
        \[
        u^2(x) = \sum_{k=0}^{n} \lambda_k \Kc_{k,4}^n(x), \quad \text{for all }x \in \inter{0}{n},
        \]
        for some $\blambda = (\lambda_0,\ldots,\lambda_{n}) \in \real^{n+1}$.
        Then the linear map
        \[
        \Phi_{\blambda} = \sum_{k=0}^{n} \lambda_k \Ec_k,
        \]
       satisfies 
       \Prop{P3}, and
        \begin{enumerate}
            \item satisfies 
            \Prop{P0}
            if $\lambda_0 = 1$;
            \item 
            satisfies 
            \Prop{P1}
            if for all
            \[
            k \in \{ \wt(\Pb)~:~ \Pb \in \supp(H) \cup \{ \I \} \},
            \]
            we have $\lambda_k \neq 0$;
            \item 
             and satisfies 
            \Prop{P2}
            for a function 
            $\epsilon(r)$
            if 
            \begin{equation}
            \label{eq:normUpperBoundbyLambdas}
             \max_{j=0,\dots,d} \| \Ec_j(H)\| \cdot \sum_{k=0}^{d} |\lambda_k^{-1}-1| \le \epsilon(r).
            \end{equation}
        \end{enumerate}
    \end{theorem}

We end this section by two remarks, highlighting how our construction of the linear map in this section connects to the existing literature.
\begin{remark}
    From \Cref{eq:krawtchouk_transform_1} and \Cref{eq:krawtchouk_transform_2}, which we prove in \Cref{prop:krawtchouk_transform}, one can immediately see that the linear subspace of maps spanned by
    $\{\Ac_k\}_{k=0}^n$
    and
    $\{\Ec_k\}_{k=0}^n$
    coincide. Moreover, we see that this subspace, which is closed under the composition of maps, is isomorphic to the \emph{Bose-Mesner algebra} associated with the \emph{Hamming scheme} 
    $H(n,4)$ \cite{delsarteAssociationSchemesCoding1998}.The Bose-Mesner algebra has been extensively studied in the literature of algebraic combinatorics and coding \cite{bannaiAlgebraicCombinatorics2021}. In the quantum information literature, this  algebraic structure 
    has been noted in the literature of quantum error correction, for example, in deriving quantum MacWilliams identities \cite{shorQuantumAnalogMacWilliams1997} and SDP bounds on quantum codes \cite{munneSDPBoundsQuantum2025}.
\end{remark}

\begin{remark}
In \cite{slotSumofSquaresHierarchiesPolynomial2022},
it was shown that in the case of commutative polynomial optimization, when applying the polynomial kernel method, using a \emph{perturbed} version of the so-called \emph{Christoffel-Darboux kernel} \cite{lasserreChristoffelDarbouxKernel2022} is advantageous. This is because in several cases, a closed form expression of this kernel exists, which can be used to establish a commutative analogue of \Prop{P3}. We refer the reader to \cite{slotSumofSquaresHierarchiesPolynomial2022, laurentOverviewConvergenceRates2026} for detailed expositions. The noncommutative analogue of the Christoffel-Darboux (CD) kernel is established in \cite{belinschiNoncommutativeChristoffelDarbouxKernels2022}.

Following the terminology in \cite[Section~3]{belinschiNoncommutativeChristoffelDarbouxKernels2022}, for $2r \leq n$, let 
$\kappa_{\tau,2r}$
be the noncommutative CD kernel 
associated with the tracial state
$\tau: \complex\langle \Xc_n \rangle \to \complex$ defined as 
\[
\tau(p) \eqdef 2^{-n} \tr[\pi(p)], \quad \forall p \in \complex\langle \Xc_n \rangle,
\]
where
$\pi$ is
the 
$*$-isomorphism defined in
\Cref{eq:representation}.
It can be shown that
\[ \kappa_{\tau, 2r} = \sum_{k=0}^{2r} \sum_{\Pb \in \Pn[k]} \pi^{-1}(\Pb) \otimes \pi^{-1}(\Pb).\]

We see that for any polynomial 
\[q(x) = \sum_{k=0}^{2r} \eta_k \Kc_{k,4}^n,\]
the Choi matrix $C_{q}$ can be expressed in terms of a \emph{perturbed} CD kernel, in the sense of \cite{slotSumofSquaresHierarchiesPolynomial2022}. Defining 
\[
\kappa_{\tau, 2r, \bbeta} \eqdef
\sum_{k=0}^{2r} \eta_k  \sum_{\Pb \in \Pn[k]} \pi^{-1}(\Pb) \otimes \pi^{-1}(\Pb).
\]
we have
\[
C_{q} = 2^{-n} ( \mathrm{T} \otimes \id) (\pi \otimes \pi) \kappa_{\tau, 2r, \bbeta},
\]
where 
$\mathrm{T}$ denotes the transposition map.
\end{remark}

\section{Proof of \Cref{thm:main}}
\label{sec:analysis_of_the_kernel}
In \Cref{thm:ev_conditions}, we showed that having a univariate polynomial 
$u \in \real[x]$
of degree at most
$r$
whose coefficients when expanded in the Krawtchouk basis satisfies certain properties, results in a linear map satisfying 
\Prop{P0}-\Prop{P3}. Hence, such a polynomial can be used to obtain rate of convergence bounds for the SOHS hierarchy over the Pauli algebra. It remains to demonstrate that such a polynomial $u$ exists and determine the exact error $\epsilon(r)$ that can be obtained from the kernel method. From this, \Cref{thm:main}
immediately follows. We conclude this section by discussing possible improvements of our bound in special cases.

In Appendix~\ref{appndx:upper_bound_lemma}, we prove the following lemma, which is analogous to \cite[Lemma~14]{slotSumofsquaresHierarchiesBinary2023}.
\begin{lemma}\label{lem:Anormbound}
    There exists a constant $\gamma_d> 0$ depending only on $d$, such that for any 
    $H=H^\dagger\in \Span(\Pn[\le d])$, 
    \begin{equation}
        \norm{ \Ec_k(H)} \le \gamma_d \norm{H}, \quad \text{for all }k \in \inter{0}{d}.
    \end{equation}
    In particular, $ \gamma_d\le(1+\sqrt2)^d$.
\end{lemma}

In view of \Cref{thm:ev_conditions}, we see from 
\Cref{lem:Anormbound} that if
\begin{equation}
    \label{eq:cond_ev_P2}
     \sum_{k=0}^{d} |\lambda_k^{-1}-1| \le \nu(r),
\end{equation}
then \Prop{P2} holds for
$
\epsilon(r) = \gamma_d \norm{H} \nu(r)
$.     
This observation together with the following result from \cite{slotSumofsquaresHierarchiesBinary2023}, paves the way for the proof of \Cref{thm:main}.

\begin{lemma}[\cite{slotSumofsquaresHierarchiesBinary2023}]
\label{lem:slotSumofsquaresHierarchiesBinary2023}
    \label{lem:classical}
    For
    $r,n \in \natural$, let 
    $\zeta_{n,r}$
    denote the smallest root of the Krawtchouk polynomial 
    $\Kc_{r,4}^n$.
    If 
    \begin{equation*}
        d(d+1) (\zeta_{n,r+1} / n )\leq 1/2,
    \end{equation*}
    then there exists a polynomial 
    $u \in \real[x]$ of degree $r$, where
    \begin{equation*}
        u^2(x) = \sum_{k=0}^{n} \lambda_k \Kc_{k,4}^n (x),
    \end{equation*}
    and $\blambda = (\lambda_0, \dots, \lambda_{n})$
    satisfies the following:
    \begin{enumerate}
        \item $\lambda_0 = 1$.
        \item For all $i \in [d]$,
        $
            1/2 \le \lambda_i \le 1
        $, and in particular, 
        $\lambda_i \neq 0$.
        \item 
        $
        \sum_{k=0}^d |\lambda_k^{-1} - 1| \le 
        2 d(d+1) (\zeta_{n, r+1} / n)
        $.
    \end{enumerate}
\end{lemma}

\begin{proof}[Proof of \Cref{thm:main}]
    Let 
    $u$ be a polynomial whose existence is guaranteed by \Cref{lem:classical}, where
    \begin{equation*}
        u^2(x) = \sum_{k=0}^{n} \lambda_k \Kc_{k,4}^n (x),
    \end{equation*}
    for
    $\blambda = (\lambda_0, \dots, \lambda_{n}) \in \real^{n+1}$.
    Then by \Cref{thm:ev_conditions},
    $\Phi_{\blambda}$ satisfies 
    \Prop{P0}-\Prop{P3}
    with 
    $\epsilon(r) =
    2\gamma_d d(d+1) \norm{H} (\frac{\zeta_{n, r+1}}{n})$. Therefore, from \Cref{prop:kernel_method} we have 
    \begin{equation*}
        H_{\min} - H_r \leq
        2\gamma_d d(d+1) \norm{H} (\zeta_{n, r+1} / n).
    \end{equation*}
\end{proof}

\begin{remark}[Improvements for structured Hamiltonians]
\label{remark:structured_Hamiltonians}
In many cases of interest, e.g. the Heisenberg model~\cite{auerbach1994}, the Hamiltonian $H$ is only composed of Pauli strings of certain weights. Using this knowledge about $H$ we can slightly improve the convergence rates in Theorem~\ref{thm:main}. Let $X\subseteq [0:d]$ and suppose $H$ only contains monomials of degree $k$ if $k \in X$. Then for any $k\not\in X$ we have $\Ec_k(H)=0$, and then
\Cref{eq:normUpperBoundbyLambdas} 
can be modified to
\begin{equation}
    \|\Phi^{-1} (H) - H\| \le \max_{j\in X} \| \Ec_j(H)\| \sum_{k \in X} |\lambda_k^{-1} - 1| \le \gamma_d \norm{H} \sum_{k \in X} |\lambda_k^{-1} - 1|.
\end{equation} This improves our bound in \Cref{eq:cond_ev_P2} by including fewer terms on the left hand side. In particular, \Cref{lem:classical} can be modified (see \cite[Lemma 7]{slotSumofsquaresHierarchiesBinary2023}) to give the tighter bound:
\begin{equation}\label{eq:lambdainvoptimization}
        \sum_{k\in X} |\lambda_k^{-1} - 1| \le 
        4 \left(\sum_{k\in X} k\right) (\zeta_{n, r+1} / n).
\end{equation}
The same assumption also leads to an improvement of our bounds in terms of the value of $\gamma_d$ appearing in \Cref{lem:Anormbound} (see Appendix \ref{sec:proofOfNormbound}). 

The extremal case for these improvements can be seen by considering $H \in \Span(\Pn[d])$. Then it is sufficient to optimize only over $\lambda_d$ both in \Cref{eq:lambdainvoptimization} and \Cref{eq:Noptfinal} (in fact, $\gamma_d=1$ trivially when we only consider one component in the harmonic decomposition), and we get 
\begin{equation}
        H_{\min} - H_r \leq 4 d \norm{H} (\zeta_{n, r + 1} / n).
    \end{equation}
\end{remark}

\section{Conclusion and future work}
In this work, we have established rate of convergence bounds for the sequence of lower approximations produced by the SOHS hierarchy for the ground state energy problem for qubit Hamiltonians. Our approach extends the polynomial kernel method~\cite{fangSumofsquaresHierarchySphere2021}, which has been used in the commutative setting for proving convergence rates of the SOS hierarchy, to the setting of optimizing a noncommutative polynomial over the Pauli algebra. By exploiting the symmetries of families of Pauli diagonal maps, we reduced the main task to finding a univariate polynomial satisfying suitable conditions. Through a suitable choice of polynomial, we then showed that the rate of convergence can be bounded in terms of the smallest root of a family of Krawtchouk polynomials. 

Our work provides the first quantitative rate of convergence bounds for SOHS relaxations of noncommutative polynomial optimization problems. Given the success of the polynomial kernel method for establishing speed of convergence rates for commutative polynomial optimization problems~\cite{laurentOverviewConvergenceRates2026}, it is natural to explore the possibility of further rate of convergence bounds for problems beyond the Pauli algebra. Another direction for the future work is the study of the exactness of the lower estimates 
$(H_r)_{r\in [n]}$ given by the SOHS hierarchy. As we previously noted, the hierarchy is guaranteed to converge at 
$r = n$. However, for the commutative analogue of the ground state energy problem it is known that the convergence can happen at a lower level. More precisely, it is shown in~\cite{sakaue2017exact}
that the SOS hierarchy for polynomial optimization over the binary cube is exact for 
\begin{equation}
\label{eq:commutative_exact_threshold}
    r \ge \frac{n + d - 1}{2}.
\end{equation}
Our asymptotic analysis suggests the possibility of a similar finite convergence at roughly $r \approx 3n/4$, as we see that for any 
$\epsilon > 0$,
\[
H_{\min}^{(n)} - H_{3n/4}^{(n)} \le \epsilon \norm{H^{(n)}},
\]
for all sufficiently large $n$. 
It remains an open problem whether a threshold similar to \Cref{eq:commutative_exact_threshold} holds in the noncommutative setting. Finally, it is worth noting the possibility of improving our bounds by exploiting different structures of the Hamiltonians such as sparsity or symmetry. In \Cref{remark:structured_Hamiltonians} we discussed some immediate improvements of our bounds in special cases. However, a full analysis of possible improvements and the tightness of our bounds would be an interesting avenue for future work.

\paragraph{Note added.} During the final preparation of the manuscript we became aware of a related work by Klep \textit{et al}, titled ``Quantitative semidefinite certificates for ground-state energies of Pauli Hamiltonians''~\cite{klep2026quantitative}, which obtains similar results. The two works were carried out independently.   

\section*{Acknowledgments}
    This project has received funding from the European Union’s Horizon Europe research and innovation programme under the project "Quantum Secure Networks Partnership" (QSNP, grant agreement No 101114043). A.A. acknowledges PhD funding from the doctoral school of Institut Polytechnique de Paris.
    C.R. is supported by France 2030 under the French
National Research Agency award number ANR-22-EXES-0013.

%%------------------------------------------------------------------------------%

% %-----------------------------------------------------------------------------%

\bibliographystyle{quantum}
\bibliography{pauli_paper}

\appendix

\section{Proof of \Cref{lem:Pauli-algebra-representation}}
\label{appndx:pauli_algebra_representation}
\begin{proof}
    For an arbitrary 
$i \in [n]$, 
consider 
$\mathscr{P}_{\circ}
\eqdef 
\mathbb{C}\langle \pauli{x}{i},  \pauli{y}{i},  \pauli{z}{i} \rangle / \mathcal{I}_\circ$, where $\mathcal{I}_{\circ}$ is the two-sided ideal generated by polynomials
\[
\pauli{x}{i}^2 - 1, 
\quad \pauli{y}{i}^2 - 1, 
\quad \pauli{z}{i}^2 - 1, 
\quad \pauli{x}{i} \pauli{y}{i} - \ib\pauli{z}{i},
\quad \pauli{y}{i} \pauli{Z}{i} - \ib\pauli{x}{i}, 
\quad \pauli{z}{i} \pauli{x}{i} - \ib\pauli{y}{i}.
\]
Using these relations, any polynomial in 
$\pauli{x}{i}, \pauli{y}{i}, \pauli{z}{i}$ 
can be reduced to a linear combination of 
$1,\pauli{x}{i},\pauli{y}{i},\pauli{z}{i}$. 
Hence 
$\dim \mathscr{P}_\circ \leq 4$. 
On the other hand, as the set
$\bounded(\complex^2) = \Span\{ \Iq, \X, \Y, \Z \}$,
the $*$-homomorphism 
$\mathscr{P}_\circ \to \bounded(\mathbb{C}^2)$ 
given by 
\[
1\mapsto \Iq,\;  \pauli{x}{i} \mapsto \X,\; \pauli{y}{i} \mapsto \Y,\; \pauli{z}{i} \mapsto \Z,
\]
is surjective, thus
$\dim \mathscr{P}_\circ = 4$. 
This implies that the above $*$-homomorphism is in fact an $*$-isomorphism between 
$\mathscr{P}_\circ$
and
$\bounded(\mathbb{C}^2)$. Consequently, we have an 
$*$-isomorphism between
$\bounded(\mathbb{C}^{2^n})$
and
$\mathscr{P}_{\circ}^{\otimes n}$.

Now consider 
$\Pcc$,
which is the algebra of complex polynomials in variables
$\bigcup_{i=1}^n \{\pauli{x}{i}, \pauli{y}{i}, \pauli{z}{i}\}$
modulo the following equations:
\begin{align*}
&(\pauli{x}{i})^2 = (\pauli{y}{i})^2 = (\pauli{z}{i})^2 = 1, \quad \text{for } i \in [n], \\
        &\pauli{x}{i}\pauli{y}{i} - \ib \pauli{z}{i} =  
        \pauli{y}{i}\pauli{z}{i} - \ib \pauli{x}{i} =
        \pauli{z}{i}\pauli{x}{i} - \ib \pauli{y}{i} = 0, \quad \text{for } i \in [n], \\
        & \pauli{A}{i} \pauli{B}{j} - \pauli{B}{j} \pauli{A}{i} = 0, \quad \text{for } i,j \in [n],\ i\neq j,\ \mathsf{A}, \mathsf{B} \in \{\mathsf{X},\mathsf{Y},\mathsf{Z}\}.
\end{align*}
It similarly follows that 
$\Pcc$ 
is spanned by the set of monomials 
$\mathsf{A}_1 \mathsf{A}_2 \cdots \mathsf{A}_n$,
where 
$\mathsf{A}_i \in \{1, \pauli{x}{i}, \pauli{y}{i}, \pauli{z}{i}\}$. 
Consider the 
$*$-homomorphism 
$\mathscr{P}_{\circ}^{\otimes n} \to \Pcc$ 
given by
\[
\mathsf{A}_1 \otimes \mathsf{A}_2 \otimes \cdots \otimes \mathsf{A}_n \mapsto \mathsf{A}_1 \mathsf{A}_2 \cdots \mathsf{A}_n,
\]
which, by what is shown above, extends to a 
$*$-homomorphism 
$\bounded(\mathbb{C}^{2^n}) \to \Pcc$. 
This latter $*$-homomorphism is surjective, and since 
$\bounded(\mathbb{C}^{2^n})$
has no nontrivial two-sided ideals \cite[Chapter XVII,~Theorem~5.2]{langAlgebra2002}, it is also injective, thus a 
$*$-isomorphism.
To complete the proof, it suffices to use the well-known fact (see e.g. \cite[Chapter XVII,~Corollary~4.6]{langAlgebra2002}) that 
$
\bounded(\mathbb{C}^{2^n})
$
has a unique irreducible representation (up to unitary equivalence). 

\end{proof}

\section{Proofs of \Cref{sec:kernel-const}}
\label{sec:apndx_proofs_kernel_const}
For 
$\mathrm{P}, \mathrm{Q} \in 
\mathcal{P}_{\!1} 
\eqdef 
\{\Iq,\X,\Y,\Z\}$, 
define 
$\eta: 
\mathcal{P}_{\!1} 
\times 
\mathcal{P}_{\!1} 
\to 
\natural_0$ 
as
\begin{equation}
    \label{eq:eta_definition}
\eta(\mathrm{P},\mathrm{Q}) 
\eqdef 
\begin{cases}
1 & \text{if } \mathrm{P}\neq \mathrm{Q} \text{ and } \mathrm{P},\mathrm{Q}\neq \Iq, \\
0 & \text{otherwise}.
\end{cases}
\end{equation}
We extend the definition of 
$\eta$ 
to 
$\Pn \times \Pn$ 
by defining 
\begin{equation}
    \label{eq:eta_extension}
    \eta(\Pb,\Qb) \eqdef \sum_{i=1}^n \eta(\mathrm{P}_i,\mathrm{Q}_i),
\end{equation}
for 
$\Pb = \mathrm{P}_{\!1} \otimes \cdots \otimes \mathrm{P}_{\!n}$ 
and 
$\Qb = \mathrm{Q}_{1} \otimes \cdots \otimes \mathrm{Q}_{n}$.
Note that for 
$\Pb, \Qb \in \Pn$, 
we have 
\begin{equation}
    \label{eq:eta_commutation_relation}
\Pb \Qb = (-1)^{\eta(\Pb,\Qb)} \Qb \Pb.
\end{equation}

An 
$n$-qubit 
\emph{Pauli diagonal map} 
is a linear map 
$\Phi: \bounded(\complex^{2^n}) \to \bounded(\complex^{2^n})$ 
acting on an 
$n$-qubit 
operator 
$X \in \bounded(\complex^{2^n})$ 
as
\begin{equation}
    \label{eq:pauli_diagonal_map}
    \Phi(X) \eqdef \sum_{\Pb \in \Pn} \alpha_{\Pb} \Pb X \Pb,
\end{equation}
for some $\alpha_{\Pb} \in \complex$. The following lemma gives an alternative form for these maps.
\begin{lemma}
    \label{lem:pauli_diagonal_alternative_form}
    A Pauli diagonal map 
    \[
    \Phi(X) = \sum_{\Qb \in \Pn} \alpha_{\Qb} \Qb X \Qb,
    \]
    can equivalently be expressed as
    \begin{equation}
        \label{eq:pauli_diagonal_alternative_form}
        \Phi(X) = 2^{-n} \sum_{\Pb \in \Pn} \beta_{\Pb} \tr(\Pb X) \Pb,
    \end{equation}
    where 
    \begin{equation}
        \label{eq:alpha_beta_relation}
        \beta_{\Pb} = \sum_{\Qb \in \Pn} (-1)^{\eta(\Pb,\Qb)} \alpha_{\Qb}, \quad \text{for } \Pb \in \Pn.
    \end{equation}
\end{lemma}
\begin{proof}
    We have 
    \begin{align*}
            \sum_{\Qb \in \Pn} \alpha_{\Qb} \Qb X \Qb 
            & =
            \sum_{\Qb \in \Pn} \alpha_{\Qb} 
            \Bigl(
            2^{-n} 
            \sum_{\Pb \in \Pn} 
            \Tr\!\bigl(
                \Qb X \Qb \Pb
                \bigr) \Pb
            \Bigr)\\ 
            &=
            2^{-n}
            \sum_{\Pb \in \Pn} 
            \bigl( \sum_{\Qb \in \Pn} \alpha_{\Qb} (-1)^{\eta(\Pb,\Qb)}\bigr)
            \tr(\Pb X) \Pb.
    \end{align*}
    where on the first line we used the fact that $\{2^{-n/2}\Pb\}_{\Pb \in \Pc_n}$ form an orthonormal basis for $\Bc(\mathbb{C}^{2^n})$ with respect to the Hilbert-Schmidt inner product and for the second line we used~\Cref{eq:eta_commutation_relation}.
\end{proof}
Recall the two families of Pauli diagonal maps, 
$\{\Ac_k\}_{k \in \inter{0}{n}}$ 
and
$\{\Ec_k\}_{k\in \inter{0}{n}}$,
defined previously in \Cref{eq:Ac_definition} and \Cref{eq:EC_def}:
\begin{equation*}
\begin{aligned}
\Ac_k(X) & \eqdef \sum_{\Qb \in \Pn[k]} \Qb X \Qb, \\
\Ec_k(X) & \eqdef 2^{-n} \sum_{\Pb \in \Pn[k]} \tr(\Pb X) \Pb.
\end{aligned}
\end{equation*}
It can be verified that any two Pauli diagonal maps commute, and in particular, the maps 
$\Ac_k$ 
and 
$\Ec_k$ 
commute for all 
$k \in \inter{0}{n}$.
The next proposition shows that
$\{\Ac_k\}_{k\in \inter{0}{n}}$
and
$\{\Ec_k\}_{k\in \inter{0}{n}}$
are in fact two different bases for a commutative subalgebra of Pauli diagonal maps.
\begin{prop}
    \label{prop:krawtchouk_transform}
    For $k \in \inter{0}{n}$, 
    \begin{align}
        \Ac_k & = \sum_{t=0}^n \mathcal K_{k,4}^n(t)\,\Ec_t,  \\
        \Ec_k & = 4^{-n}\sum_{t=0}^n \mathcal K_{k,4}^n(t)\,\Ac_t, 
    \end{align}
    where 
    $\mathcal K_{k,4}^n$ 
    denotes the 
    $4$-ary 
    Krawtchouk polynomial of degree 
    $k$, defined in \Cref{eq:krawtchouk_definition}. In particular,
    \begin{equation}
        \Span \{ \Ac_0, \Ac_1, \dots, \Ac_n\}  = \Span \{ \Ec_0, \Ec_1, \dots, \Ec_n\}.
    \end{equation}
\end{prop}
To prove \Cref{prop:krawtchouk_transform}, we need the following lemma.
\begin{lemma}
    \label{lem:krawtchouk_transform}
    For $k,t \in \inter{0}{n}$, 
    \begin{equation}
        \label{eq:krawtchouk_transform_lemma}
        \sum_{\Qb \in \Pn[k]} (-1)^{\eta(\Pb,\Qb)} = \mathcal K_{k,4}^n(t), \quad \text{for any } \Pb \in \Pn[t].
    \end{equation}
\end{lemma}
\begin{proof}[Proof of \Cref{lem:krawtchouk_transform}]
    Without loss of generality, assume that $\Pb$ has all its non-identity single-qubit Pauli operators appearing in the first $t$ positions, i.e., $\Pb = \mathrm{P}_1 \otimes \cdots \otimes \mathrm{P}_t \otimes \Iq^{\otimes (n-t)}$ for some $\mathrm{P}_1,\ldots,\mathrm{P}_t \in \{\X,\Y,\Z\}$.
    We rewrite the left-hand side of \Cref{eq:krawtchouk_transform_lemma} as
     \begin{align}
        \sum_{\Qb \in \Pn[k]}
        (-1)^{\eta(\Pb,\Qb)} 
        & = 
        \sum_{\Qb \in \Pn[k]} 
        \prod_{j=1}^n 
        (-1)^{\eta(\mathrm{P}_j,\mathrm{Q}_j)}
     \end{align}
     Now let us partition the sum over 
     $\Qb \in \Pn[k]$ 
     according to the number of positions at which 
     $\Qb$ 
     and 
     $\Pb$ are both non-identity. Let $i$ index the number of such positions, then we are left to sum over the choice of positions $\{m_1, \dots, m_i\} \subseteq [t]$ and $\{m_1', \dots, m_{k-i}'\} \subseteq [t+1:n]$ where these non-trivial overlaps occur. Finally we sum over the choice of Pauli operators on these nontrivial overlaps, overall this expansion leads to the summation
     \begin{equation}
        \sum_{\Qb \in \Pn[k]}
        (-1)^{\eta(\Pb,\Qb)} =
        \sum_{i = 0}^{k} 
        \sum_{
            \substack{                
                \{m_1,\ldots,m_i\} \subseteq [t] \\
                \{m'_1,\ldots,m'_{k-i}\} \subseteq \inter{t+1}{n}
            }
        }
        \sum_{
            \substack{\mathrm{Q}_{m_1},\ldots,\mathrm{Q}_{m_i} \in \{\X,\Y,\Z\} \\ 
            \mathrm{Q}_{m'_1},\ldots,\mathrm{Q}_{m'_{k-i}} \in \{\X,\Y,\Z\}}
        }
        \prod_{u=1}^i 
        (-1)^{\eta(\mathrm{P}_{m_u},\mathrm{Q}_{m_u})} 
        \prod_{v=1}^{k-i} 
        (-1)^{
            \eta(\Iq,\mathrm{Q}_{m'_{v}})
            }\,.
     \end{equation}

    Noting that $(-1)^{\eta(\Iq,\mathrm{Q}_{m'_{v}})} = 1$ always, we have
     \begin{align*}
        \sum_{\Qb \in \Pn[k]} (-1)^{\eta(\Pb,\Qb)} 
        & = 
        \sum_{i = 0}^{k} 
        \sum_{             
                \{m_1,\ldots,m_i\} \subseteq [t] 
            }
        \sum_{
            \mathrm{Q}_{m_1},\ldots,\mathrm{Q}_{m_i} \in \{\X,\Y,\Z\} 
        }
        \prod_{u=1}^i (-1)^{\eta(\mathrm{P}_{m_u},\mathrm{Q}_{m_u})} 
        \left(
        \sum_{
            \substack{
            m'_1,\ldots,m'_{k-i} \subseteq \inter{t+1}{n}\\
            \mathrm{Q}_{m'_{1}},\ldots,\mathrm{Q}_{m'_{k-i}} \in \{\X,\Y,\Z\}
            }
        }
        1 
        \right)
        \\
        & =
        \sum_{i = 0}^{k} 
        \sum_{             
            \{m_1,\ldots,m_i\} \subseteq [t]
        }
        \sum_{
            \mathrm{Q}_{m_1},\ldots,\mathrm{Q}_{m_i} \in \{\X,\Y,\Z\} 
        } 
        \prod_{u=1}^i (-1)^{\eta(\mathrm{P}_{m_u},\mathrm{Q}_{m_u})} 
        3^{k-i}\binom{n-t}{k-i} \\
        & =
        \sum_{i = 0}^{k}
        3^{k-i}\binom{n-t}{k-i}
        \sum_{
            \{m_1,\ldots,m_i\} \subseteq [t]
        }
        \prod_{u=1}^i 
        \left(
            \sum_{
                \mathrm{Q}_{m_u} \in \{\X,\Y,\Z\}
            }
            (-1)^{\eta(\mathrm{P}_{m_u},\mathrm{Q}_{m_u})}
        \right)\\
        & = 
        \sum_{i = 0}^{k}
        3^{k-i}\binom{n-t}{k-i}
            \sum_{
                \{m_1,\ldots,m_i\} \subseteq [t]
            }
                \prod_{u=1}^i
                \left(
                    (-1) + (-1) + 1
                \right) \label{eq:mones_sum} \\
        & =
        \sum_{i = 0}^{k}
        (-1)^i 3^{k-i}\binom{t}{i}\binom{n-t}{k-i} \\
        & = \mathcal K_{k,4}^n(t),
     \end{align*}
     where the fourth line follows from the fact that
     exactly one 
     $\mathrm{P}_{m_u} \in \{\X,\Y,\Z\}$
     commutes with 
    $\mathrm{Q}_{m_u} \neq \Iq$,
    and the other two anticommute.
\end{proof}
\begin{proof}[Proof of \Cref{prop:krawtchouk_transform}]
We first prove \Cref{eq:krawtchouk_transform_1}. For $X \in \bounded(\complex^{2^n})$, we have
\begin{align*}
\Ac_k(X) 
&\;= \sum_{\Qb \in \Pn[k]} \Qb X \Qb \\
&\overset{\tikzrefsize{\textsc{Lem}~\ref{lem:pauli_diagonal_alternative_form}}}{=} 
2^{-n} 
\sum_{\Pb \in \Pn} \Bigl(\sum_{\Qb \in \Pn[k]} (-1)^{\eta(\Pb,\Qb)}\Bigr) \tr(\Pb X) \Pb\\
&\;= 2^{-n}
\sum_{t=0}^n
\sum_{\Pb \in \Pn[t]} 
\left(\sum_{\Qb \in \Pn[k]} 
(-1)^{\eta(\Pb,\Qb)}\right)
\tr(\Pb X) \Pb \\
&\overset{\tikzrefsize{\textsc{Lem}~\ref{lem:krawtchouk_transform}}}{=} 
\sum_{t=0}^n
\mathcal K_{k,4}^n(t)
2^{-n} 
\sum_{\Pb \in \Pn[t]}
\tr(\Pb X) \Pb \\
&\;= 
\sum_{t=0}^n \mathcal K_{k,4}^n(t)\,\Ec_t(X).
\end{align*}
We now prove \Cref{eq:krawtchouk_transform_2}. From \Cref{eq:krawtchouk_transform_1}, proven above, we have
\begin{align*}
    \sum_{t=0}^n \mathcal K_{k,4}^n(t)\,\Ac_t
&=
\sum_{t=0}^n \mathcal K_{k,4}^n(t) 
\left(\sum_{t'=0}^n \Kc_{t,4}^n(t') \Ec_{t'}\right) \\
&= \sum_{t'=0}^n 
\left( \sum_{t=0}^n \mathcal K_{k,4}^n(t) \Kc_{t,4}^n(t') \right) 
\Ec_{t'} \\
&=
\sum_{t'=0}^n 4^n \delta_{t',k} \Ec_{t'} \\
&= 4^n \Ec_k\,,
\end{align*}
where on the penultimate line we used the orthogonality relation for Krawtchouk polynomials from~\Cref{eq:krawtchouk_orthogonality_implied}.
\end{proof}
\begin{prop}
    \label{prop:choi_PDMaps}
    For $k \in \inter{0}{n}$, the Choi matrices of $\Ac_k$ and $\Ec_k$ are given by
    \begin{align*}
        J_{\Ac_k} & = \sum_{\Pb \in \Pn[k]} \vec{P} \cev{P}, \\
        J_{\Ec_k} & = 2^{-n} \sum_{\Pb \in \Pn[k]} \Pb^\top \otimes \Pb.
    \end{align*}
\end{prop}
\begin{proof}
    By the definition of the Choi matrix and the map 
    $\Ac_k$,
    we have
    \begin{align*}
        \choi{\Ac_k} 
        & =
         \sum_{i,j \in \{0,1\}^n} \ketbra{i}{j} \otimes \Ac_k(\ketbra{i}{j}) \\
        & =
         \sum_{i,j \in \{0,1\}^n} \ketbra{i}{j} \otimes \Bigl(\sum_{\Pb \in \Pn[k]} \Pb \ketbra{i}{j} \Pb\Bigr)\\
        & = \sum_{\Pb \in \Pn[k]} \left(
            \sum_{i \in \{0,1\}^n} \ket{i} \otimes \Pb \ket{i} 
        \right)
        \left(
            \sum_{j \in \{0,1\}^n} \bra{j} \otimes \bra{j} \Pb  
        \right) \\
        & =
         \sum_{\Pb \in \Pn[k]} \vec{P} \cev{P}.
    \end{align*}
    Similarly,
    \begin{align*}
        \choi{\Ec_k}
        & =
            \sum_{i,j \in \{0,1\}^n} \ketbra{i}{j} \otimes \Ec_k(\ketbra{i}{j}) \\
        & =
            \sum_{i,j \in \{0,1\}^n} \ketbra{i}{j} \otimes \Bigl(2^{-n} \sum_{\Pb \in \Pn[k]} \tr(\Pb \ketbra{i}{j}) \Pb\Bigr) \\
        & = 2^{-n} \sum_{\Pb \in \Pn[k]} 
        \left(
            \sum_{i,j \in \{0,1\}^n} 
        \bra{j} \Pb \ket{i}
        \ketbra{i}{j}
        \right) \otimes \Pb \\
        & = 2^{-n} \sum_{\Pb \in \Pn[k]} \Pb^\top \otimes \Pb.
    \end{align*}
\end{proof}
\begin{prop}
\label{prop:orthogonal_projection}
    For 
    $ k \in \inter{0}{n}$,
    let
    $
    \widehat{J}_{\Ac_k}
    $
    be the operators defined in 
    \Cref{eq:normalized_choi_Ac}.
    The family of operators
    $\{\widehat{J}_{\Ac_k}\}_{k \in \inter{0}{n}}$
    satisfy
    \[
    \widehat{J}_{\Ac_k}^\dagger = \widehat{J}_{\Ac_k}, 
    \qquad 
    \widehat{J}_{\Ac_{k'}}\widehat{J}_{\Ac_k} = 
    \delta_{k,k'}\widehat{J}_{\Ac_k}\,,
    \qquad
    \sum_{k=0}^n \widehat{J}_{\Ac_k} = \I \otimes \I.
    \]
\end{prop}
\begin{proof}
    From \Cref{prop:choi_PDMaps}, we know 
    $
    J_{\Ac_k} = \sum_{\Pb \in \Pn[k]} \vec{P} \cev{P}
    $, which immediately implies the Hermiticity of 
    $\widehat{J}_{\Ac_k}= 2^{-n} J_{\Ac_k}$.
    Furthermore, from \Cref{eq:krawtchouk_transform_choi_1}, we have
    
    \begin{align*}
    \sum_{k=0}^n \choi{\Ac_k} 
    & =
    \sum_{k=0}^n \sum_{t=0}^n \mathcal K_{k,4}^n(t)\,\choi{\Ec_t} = \sum_{t=0}^n {\choi{\Ec_t}} \left(\sum_{k=0}^n \mathcal K_{k,4}^n(t)\right).
\end{align*}
Using 
\Cref{eq:krawtchouk_orthogonality_special}, we get
\begin{align*}
    \sum_{k=0}^n \choi{\Ac_k}=
    \sum_{t=0}^n 4^n {\choi{\Ec_t}} \delta_{t,0} = 4^n \choi{\Ec_0}.
\end{align*}
By \Cref{prop:choi_PDMaps}, $\choi{\Ec_0} = 2^{-n}(\I \otimes \I)$,
which implies that 
\[
\sum_{k=0}^n \widehat{J}_{\Ac_k} = 2^{-n} \sum_{k=0}^n \choi{\Ac_k} = \I \otimes \I.
\]
To show orthogonality, first observe that 
for any 
$\Pb, \Qb \in \Pn$,
\begin{align*}
    \cevvec{\Pb}{\Qb} =
    \left( 
    \sum_{i \in \{0,1\}^n} \bra{i} \otimes \bra{i} \Pb
    \right)
    \left( 
    \sum_{j \in \{0,1\}^n} \ket{j} \otimes \Qb \ket{j} 
    \right) 
    = \sum_{i \in \{0,1\}^n} \langle i | \Pb \Qb | i \rangle 
    = \tr(\Pb \Qb)
    = 2^n \delta_{\Pb, \Qb}.
\end{align*}
This implies
\begin{align*}
    \widehat{J}_{\Ac_{k'}}\widehat{J}_{\Ac_k} &= 
    4^{-n} \choi{\Ac_{k'}} \choi{\Ac_{k}}\\
    &=
    4^{-n} \sum_{\Pb \in \Pn[k']} \vec{\Pb}\cev{\Pb} \sum_{\Qb \in \Pn[k]} \vec{\Qb}\cev{\Qb} \\
    &=
    4^{-n} \sum_{\Pb \in \Pn[k']} \sum_{\Qb \in \Pn[k]} \cev{\Pb}\vec{\Qb} \vec{\Pb} \cev{\Qb} \\
    &=
    4^{-n} \sum_{\Pb \in \Pn[k']} \sum_{\Qb \in \Pn[k]} 2^n \delta_{\Pb, \Qb} \vec{\Pb} \cev{\Qb}\\
    &=
    2^{-n} \delta_{k,k'} \choi{\Ac_k}\\
    &= \delta_{k,k'} \widehat{J}_{\Ac_{k}}.
\end{align*}
\end{proof}
\section{Proof of \Cref{lem:Anormbound}}\label{sec:proofOfNormbound}
\label{appndx:upper_bound_lemma}
\begin{proof}
Notice that for self-adjoint $H=\sum_{\Pb\in\mathcal{P}_{n,\le d}}h_{\Pb} \Pb\in \Span(\mathcal{P}_{n,\le d})$, we have $H=\sum_{i=0}^d \Ec_i(H)$.
We define the optimizations
\begin{equation}\label{eq:optN}
\begin{split}
    N(n,d,k)&\eqdef\sup_{H=H^\dagger \in\Span(\mathcal{P}_{n,\le d})}\{\norm{\Ec_k(H)}: \quad \norm{H}\le 1 \}\\
&=\sup_{\norm{\psi}=1}\quad\sup_{H=H^\dagger\in\Span(\mathcal{P}_{n,\le d})}\{\lvert\bra{\psi}\Ec_k(H)\ket{\psi}\rvert: \quad \norm{H}\le 1 \}
\end{split}
\end{equation}
 and 
\begin{equation}
    N(n,d)\eqdef\sup_{k} N(n,d,k)\,.
\end{equation}
Our aim is to prove that for any $n,d$ we have $N(n,d)\le \gamma_d$ for some $d$-dependent (and most importantly $n$-independent) constant $\gamma_d$. To do so, first we write for the objective function that
\begin{equation}\label{eq:Akobjective}
\begin{split}
   \lvert\bra{\psi}\Ec_k(H)\ket{\psi}\rvert
&=\left\lvert\sum_{\substack{\Pb\in\Pn[k]}}h_{\Pb}
   \bra{\psi}\Pb\ket{\psi}\right\rvert
   \eqdef\lvert h_k\rvert.
\end{split}
\end{equation}
We then look to relax the constraint $\|H\| \le 1$. For $x\in [-1,1]$, let 
\begin{equation}
    \Delta_x(A)\eqdef xA + \frac{(1-x)}{2}\Tr (A)\Iq    
\end{equation}
be the one-qubit depolarizing map.
Note that this is a unital Pauli diagonal map sending any non-identity Pauli $\sigma$ to $x\sigma$ and 
$\norm{\Delta_x}_{\infty\to\infty}=1$. To see the latter, note that the depolarization channel is completely positive for $x\ge 0$ (therefore by unitality $\norm{\Delta_x}_{\infty\to\infty}=1$), while for $x<0$ we write
\begin{equation}
    \Delta_{x}(A)=\Delta_{-1}\left(\Delta_{-x}(A)\right)
\end{equation}
and note that
\begin{equation}
    \Delta_{-1}(A)= \Y A^\top\Y
\end{equation}
is an operator norm isometry, and so is $\Delta_{-1}^{\otimes n}$.
Then, assuming $x< 0$ we have
\begin{equation}
    \norm{\Delta_{x}^{\otimes n}(A)}=\norm{\Y^{\otimes n} \Delta_{-x}^{\otimes n}(A)^\top\Y^{\otimes n}}=\norm{\Delta_{-x}^{\otimes n}\left( A\right)}\le
    \norm{\Delta_{-x}^{\otimes n}}_{\infty\to\infty}\norm{A}=\norm{A}
\end{equation}
for any $n$-qubit operator $A$. Equality is achieved by the identity operator.

Then for any unit vector $\ket{\psi}$ the map $\bra{\psi} \Delta^{\otimes n}_x(\cdot) \ket{\psi}$ also has norm $\norm{\bra{\psi} \Delta^{\otimes n}_x(\cdot) \ket{\psi}}_{\infty\to\infty}=1$, so we can relax the constraint of the optimization in the second line of \Cref{eq:optN} into $\lvert
    \sum_{i=0}^d h_i x^i
    \rvert \le 1$ for all $x \in [-1,1]$ as
\begin{equation}
    \|H\|\ge\abs{\bra{\psi} \Delta_x^{\otimes n}(H) \ket{\psi}}
    = \left\lvert\sum_{\Pb\in\mathcal{P}_n}h_{\Pb} x^{\wt(\Pb)}\bra{\psi}\Pb\ket{\psi}
    \right\rvert
    =
    \left\lvert
    \sum_{i=0}^d \left(
    \sum_{\substack{\Pb\in\Pn[i]}}
    h_{\Pb}
    \bra{\psi}\Pb\ket{\psi}
    \right)
    x^i
    \right\rvert
    =
    \left\lvert
    \sum_{i=0}^d h_i x^i
    \right\rvert
    .
\end{equation}
This results in the upper bound
\begin{equation}\label{eq:Noptfinal}
     N(n,d,k)\le
     \sup_{\{h_i\}_{i\in[0:d]}}\left\{\lvert h_k\rvert: \quad \left\lvert\sum_{i=0}^d h_ix^i\right\rvert \le 1 \quad \forall x\in[-1,1] \right\},
\end{equation}
where on the right hand side we see the same $n$-independent optimization problem as in \cite[Equation 45]{slotSumofsquaresHierarchiesBinary2023}. Then optimizing over $k$, we get an $n$-independent bound for $N(n,d)$. These optimizations were performed and the optimal values of $\gamma_d$ were also determined in \cite{slotSumofsquaresHierarchiesBinary2023}. They provide the solution in terms of the elements of Chebyshev polynomials, in particular, they give the estimate
\begin{equation}
    \gamma_d\coloneqq N(n,d) \le(1+\sqrt2)^d.
\end{equation}

\end{proof}

\end{document}